%% file: Paper_Part_I.tex
\documentclass[twocolumn]{IEEEtran}

\usepackage{multirow}
\usepackage{graphicx}
\usepackage{xcolor}
\usepackage{epstopdf}
\usepackage{caption}
\usepackage{subcaption}
\usepackage{amsmath}
\usepackage{bm}
\usepackage{pgfplots}
\usepackage{footnote}
\usepackage{stfloats}
\usepackage{placeins}
\usepackage{color,soul}
\usepackage{dsfont}
\usepackage{amssymb}
\usepackage[boxruled]{algorithm2e}
\usepackage{pbox}
\usepackage{enumerate}

\usepackage{etextools}
\usepackage{mathtools}

\usepackage[hidelinks]{hyperref}

\usepackage{amsthm}

\theoremstyle{plain}
\newtheorem{thm}{Theorem} 

\mathtoolsset{showonlyrefs}

\newtheorem{lemma}{Lemma}

\theoremstyle{definition}
\newtheorem{defn}[thm]{Definition} 
\newtheorem{exmp}[thm]{Example} 

\newlength
\figureheight
\newlength
\figurewidth

\def\x{{\mathbf x}}
\def\X{{\mathbf X}}

\def\d{{\mathbf d}}

\def\xi{\x_i}

\def\Y{{\mathbf Y}}
\def\z{{\mathbf z}}

\def\S{{\mathbf S}}
\def\G{{\mathbf G}}
\def\D{{\mathbf D}}

\def\A{{\mathbf A}}

\def\alfa{{\boldsymbol \alpha}}
\def\gama{{\boldsymbol \gamma}}
\def\Gama{{\boldsymbol \Gamma}}
\def\Delt{{\boldsymbol \Delta}}
\def\delt{{\boldsymbol \delta}}
\def\O{{\boldsymbol \Omega}}

\def\Poi{{P_{0,\infty}}}
\def\Loi{{\ell_{0,\infty}}}

\def\C{{\mathbf C}}

\def\R{{\mathbf R}}

\def\E{{\mathbf E}}

\DeclarePairedDelimiterX\set[1]\lbrace\rbrace{\def\given{\;\delimsize\vert\;}#1}

\begin{document}

\title{Working Locally Thinking Globally - Part I: Theoretical Guarantees for Convolutional Sparse Coding}
\author{Vardan~Papyan*,
	Jeremias~Sulam*,
	Michael~Elad

\thanks{*The authors contributed equally to this work. 
	
All authors are with the Computer Science Department, the Technion - Israel Institute of Technology.}}

\maketitle

\begin{abstract}
The celebrated sparse representation model has led to remarkable results in various signal processing tasks in the last decade.
However, \mbox{despite} its initial purpose of serving as a global prior for entire signals, it has been commonly used for modeling low dimensional patches due to the computational constraints it entails when deployed with learned dictionaries.
A way around this problem has been proposed recently, adopting a convolutional sparse representation model.
This approach assumes that the global dictionary is a concatenation of banded Circulant matrices.
Although several works have presented algorithmic solutions to the global pursuit problem under this new model, very few truly-effective guarantees are known for the success of such methods.
In the first of this two-part work, we address the theoretical aspects of the sparse convolutional model, providing the first meaningful answers to corresponding questions of uniqueness of solutions and success of pursuit algorithms.
To this end, we generalize mathematical quantities, such as the $\ell_0$ norm, the mutual coherence and the Spark, to their counterparts in the convolutional setting, which intrinsically capture local measures of the global model.
In a companion paper, we extend the analysis to a noisy regime, addressing the stability of the sparsest solutions and pursuit algorithms, and demonstrate practical approaches for solving the global pursuit problem via simple local processing.
\end{abstract}

\begin{IEEEkeywords}
	Sparse Representations, Convolutional Sparse Coding, Uniqueness Guarantees, Orthogonal Matching Pursuit, Basis Pursuit, Global modeling, Local Processing.
\end{IEEEkeywords}

\section{Introduction}

A popular choice for a signal model, which has proven to be very effective in a wide range of applications, is the celebrated sparse representation prior \cite{Bruckstein2009,Mairal2014,Romano2014,Dong2011}. In this framework, one assumes a signal $\X \in \mathbb{R}^N$ to be a sparse combination of a few columns (atoms) $\d_i$ from a collection $\D \in \mathbb{R}^{N\times M}$, termed dictionary. In other words, $\X = \D\Gama$ where $\Gama\in \mathbb{R}^{M}$ is a sparse vector. Finding such a vector can be formulated as the following optimization problem:
\begin{equation}
\min_{\Gama}\ g(\Gama)\ \text{ s.t. } \D\Gama = \X,
\label{Eq:Global Sparse Model}
\end{equation}
where $g(\cdot)$ is a function which penalizes dense solutions, such as the $\ell_1$ or $\ell_0$ ``norms''\footnote{Despite the $\ell_0$ not being a norm (as it does not satisfy the homogeneity property), we will use this jargon throughout this paper for the sake of simplicity.}. For many years, analytically defined matrices or operators were used as the dictionary $\D$ \cite{Mallat1993,Elad2005}. However, designing a model from real examples by some learning procedure has proven to be more effective, providing sparser solutions \cite{Engan1999,Aharon2006,Mairal2009a}. This led to vast work that deploys dictionary learning in a variety of applications \cite{Li2011,Zhang2010,Dong2011,Gao2012,Yang2014}.

Generally, solving a pursuit problem is a computationally challenging task. As a consequence, most such recent successful methods have been deployed on relatively small dimensional signals, commonly referred to as \emph{patches}. Under this \emph{local} paradigm, the signal is broken into overlapped blocks and the above defined sparse coding problem is reformulated as
\begin{equation}
\forall \ i \quad \min_{\alfa}\ g(\alfa)\ \text{ s.t. }\ \D_L\alfa = \R_i\X,
\end{equation}
where $\D_L \in \mathbb{R}^{n\times m}$ is a local dictionary, and $\R_i\in \mathbb{R}^{n\times N}$ is an operator which extracts a small local patch of length $n\ll N$ from the global signal $\X$. In this set-up, one processes each patch independently and then aggregates the estimated results using plain averaging in order to recover the global reconstructed signal. A local-global gap naturally arises when solving global tasks with this local approach. The reader is referred to \cite{Sulam2015,Romano2015b,Romano2015,Papyan2016,batenkov2017global,Zoran2011} for further insights on this dichotomy.

The above discussion suggests that in order to find a consistent global representation for the signal, one should propose a global sparse model. However, employing a general global dictionary is infeasible due to the curse of dimensionality and the complexity involved. An alternative is a global model in which the signal is composed as a superposition of local atoms. The family of dictionaries giving rise to such signals is a concatenation of banded Circulant matrices. This global model benefits from having a local shift invariant structure -- a popular assumption in signal and image processing -- suggesting an interesting connection to the above-mentioned local modeling.

When the dictionary $\D$ has this structure of a concatenation of banded Circulant matrices, the pursuit problem in \eqref{Eq:Global Sparse Model} is usually known as convolutional sparse coding \cite{Grosse2007}. Recently, several works have addressed the problem of using and training such a model in the context of image inpainting, super-resolution, and general image representation \cite{Bristow2013,Heide2015,Kong2014,Wohlberg2014,Gu2015}. These methods exploit an ADMM formulation \cite{Boyd2011} in the Fourier domain in order to search for the sparse codes and train the dictionary involved. Several variations have been proposed for solving the pursuit problem, yet there has been no theoretical analysis of their success. 

In this two-part work, we consider the following set of questions: Assume a signal $\X$ is created by multiplying a sparse vector $\Gama$ by a global structured dictionary $\D$ that consists of a union of banded and Circulant matrices. Then, 
\begin{enumerate}
	\item Can we guarantee the uniqueness of such a global (convolutional) sparse vector?
	\item Can global pursuit algorithms, such as the ones suggested in recent works, be guaranteed to find the true underlying sparse code, and if so, under which conditions?
	\item Can we guarantee a stability of the sparse approximation problem, and a stability of corresponding pursuit methods in a noisy regime?; And
	\item Can we solve the global pursuit by restricting the process to local pursuit operations?
\end{enumerate}
In part I of our work, we focus on answering questions 1 and 2, while questions 3 and 4 are analyzed and answered in the sequel paper.

A na\"ive approach to address such theoretical questions is to apply the fairly extensive results for sparse representation and compressed sensing to the above defined model \cite{Elad_Book}. However, as we will show throughout this paper, this strategy provides nearly useless results and bounds from a global perspective. Therefore, there exists a true need for a deeper and alternative analysis of the sparse coding problem in the convolutional case which would yield meaningful bounds. 

In this work, we will demonstrate the futility of the $\ell_0$-norm in capturing the concept of sparsity in the convolutional model. This, in turn, motivates us to propose a new localized measure -- the $\Loi$ norm. Based on it, we redefine our pursuit into a problem that operates locally while thinking globally. To analyze this problem, we extend useful concepts, such as the Spark and mutual coherence, to the convolutional setting. We then provide claims for uniqueness of solutions and for the success of pursuit methods in the noiseless case, both for greedy algorithms and convex relaxations. These are the main concerns of this part and will pave the theoretical foundations for part II, where we will extend the analysis to a more practical scenario of handling noisy data.


This paper is organized as follows. We begin by reviewing the unconstrained global (traditional) sparse representation model in Section \ref{Sec:Preliminaries}, followed by a detailed description of the convolutional structure in Section \ref{Sec:Conv_Model}. Section \ref{Sec:Global2Local} briefly motivates the need of a thorough analysis of this model, which is then provided in Section \ref{Sec:TheoStudy}. We introduce additional mathematical tools in Section \ref{Sec:ShiftedMutualCoherence}, which provide further insight into the convolutional model. 
Finally, we conclude and motivate the next part of this work in Section \ref{Sec:Conclusions}.

\section{The Global Sparse Model -- Preliminaries}
\label{Sec:Preliminaries}

The by-now classic sparse representation model assumes a signal $\X \in \mathbb{R}^N$ can be expressed as $\X = \D \Gama$, where $\D \in \mathbb{R}^{N\times M}$, $\Gama \in \mathbb{R}^M$ and $\|\Gama\|_0 \ll N$. In this last expression, the $\ell_0$ pseudo-norm $\| \cdot \|_0$ counts the non-zero elements in its argument. Finding the sparsest vector for a given signal is known as Sparse Coding, and it attempts to solve the constrained $P_0$ problem:
\begin{equation}
(P_0): \quad \underset{\Gama}{\min}\ \|\Gama\|_0 \ \text{ s.t. } \  \D\Gama = \X. 
\label{Eq:P0problem}
\end{equation}
Several results have shed light on the theoretical aspects of this problem, claiming a unique solution under certain circumstances. These guarantees are given in terms of properties of the dictionary $\D$, such as the \emph{Spark}, defined as the minimum number of linearly dependent columns (atoms) in $\D$ \cite{Donoho2003}. Formally,
\begin{equation}
\sigma(\D) = \underset{\Gama}{\min}\ \|\Gama\|_0 \ \text{ s.t. } \ \D\Gama = \mathbf{0}, \ \Gama \neq \mathbf{0}.
\end{equation}
Based on this property, a solution obeying $\|\Gama\|_0 < \sigma(\D)/2$ is necessarily the sparsest one \cite{Donoho2003}. Unfortunately, this bound is of little practical use, as computing the Spark of a matrix is a combinatorial problem -- and infeasible in practice.  

Another guarantee is given in terms of the \emph{mutual coherence} of the dictionary, $\mu(\D)$. This measure quantifies the similarity of atoms in the dictionary, defined in \cite{Donoho2003} as:
\begin{equation}
\mu(\D) = \underset{i\neq j}{\max} \frac{|\d_i^T \d_j|}{\|\d_i\|_2\cdot\|\d_j\|_2}.
\end{equation}
Hereafter, we will assume without loss of generality that all atoms in $\D$ are normalized to unit $\ell_2$ norm. A relation between the Spark and the mutual coherence was shown in \cite{Donoho2003}, stating that $\sigma(\D) \geq 1 + \frac{1}{\mu(\D)}$. This, in turn, enables the formulation of a practical uniqueness bound guaranteeing that $\Gama$ is the unique solution of the $P_0$ problem if $\|\Gama\|_0 < \frac{1}{2} \left(1 + 1/\mu(\D)\right)$.

Solving the $P_0$ problem is NP-hard in general. Nevertheless, its solution can be approximated by either greedy pursuit algorithms, such as the Orthogonal Matching Pursuit (OMP) \cite{Pati1993a,Chen1989}, or convex relaxation approaches like Basis Pursuit (BP) \cite{Chen2001}. Despite the difficulty of this problem, these methods (and other similar ones) have been proven to recover the true solution if $\|\Gama\|_0 < \frac{1}{2} \left(1 + 1/\mu(\D)\right)$ \cite{Donoho2006,Tropp2004,Donoho2003,Gribonval2003}.

In real world applications, due to noisy measurements and model imperfections, the idealistic setting portrayed above is not directly applicable. Nevertheless, one can extend the model to include signal perturbations, obtaining the following problem:
\begin{equation}
(P_0^\epsilon): \quad \underset{\Gama}{\min}\ \|\Gama\|_0 \ \text{ s.t. } \|\D\Gama - \Y\|_2 \leq \epsilon.
\end{equation}
We will deffer studying this case here, as it will be analyzed in detail in part two of this report.

\section{The Convolutional Sparse Model}
\label{Sec:Conv_Model}
When handling large dimensional signals, using an unstructured dictionary becomes unfeasible. 
In this section we will enforce a constraint on the global dictionary, resulting in both theoretical and practical benefits.

\begin{figure*}[t]
	\centering
	\includegraphics[trim = -40 60 40 30 ,width=0.9\textwidth]{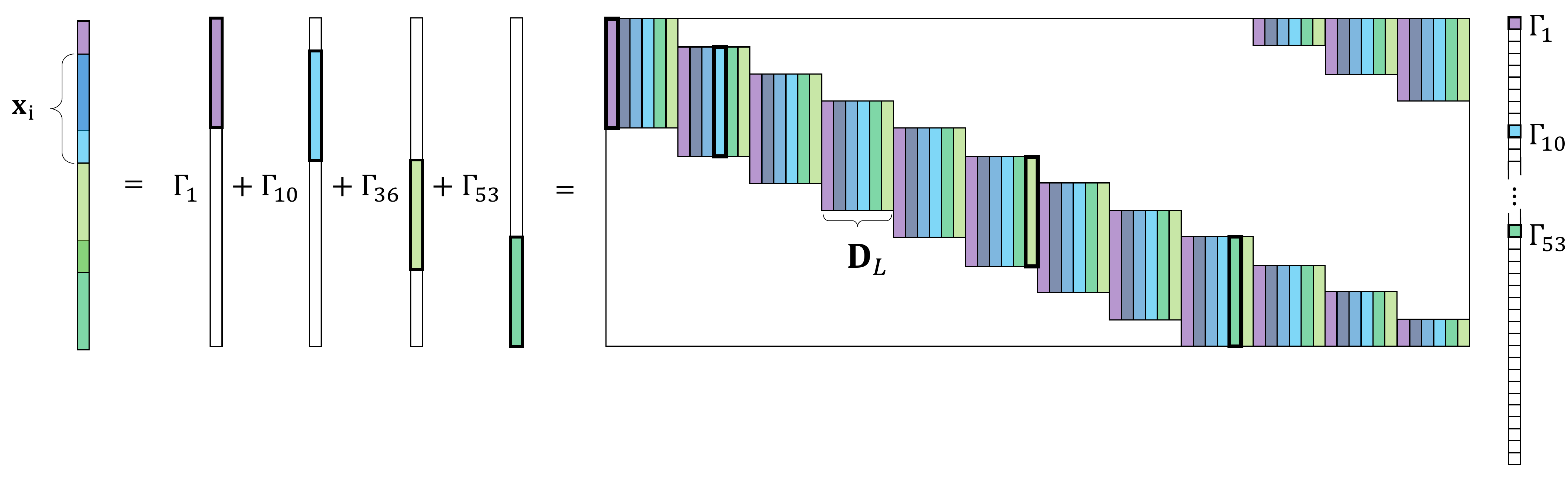}
	\caption{The convolutional model description, and its composition in terms of the local dictionary $\D_L$.}
	\label{GlobalDiagram}
\end{figure*}

Consider the global dictionary to be a concatenation of $m$ banded Circulant matrices\footnote{The choice of Circulant matrices comes to alleviate boundary problems.}, where each such matrix has a band of width $n \ll N$. As such, by simple permutation of its columns, such a dictionary consists of all shifted versions of a \emph{local} dictionary $\D_L$ of size $n\times m$. This model is commonly known as Convolutional Sparse Representation \cite{Grosse2007,Bristow2014,Heide2015}. Hereafter, whenever we refer to the global dictionary $\D$, we assume it has this structure. Assume a signal $\X$ to be generated as $\D\Gama$. In Figure \ref{GlobalDiagram} we describe such a global signal, its corresponding dictionary that is of size $N \times mN$ and its sparse representation, of length $mN$. We note that $\Gama$ is built of $N$ distinct and independent sparse parts, each of length $m$, which we will refer to as the local sparse vectors $\alfa_i$.
In this section we shall propose several different interpretations of signals emerging from this model, and as we shall see, these will serve us well in the later analysis.
 
Consider a sub-system of equations extracted from $\X=\D\Gama$, by multiplying this system by the patch extraction operator $\R_i$. The resulting system is $\x_i$ = $\R_i \X = \R_i \D \Gama$, where $\x_i$ is a patch of length $n$ extracted from $\X$ from location $i$. Observe that in the set of rows extracted, $\R_i \D$, there are only $(2n-1)m$ columns that are non-trivially zero. Define the operator $\S_i\in \mathbb{R}^{(2n-1)m \times mN}$ as a columns' selection operator, such that $\R_i \D \S_i^T$ preserves all the non-zero columns in $\R_i \D$. Thus, the subset of equations we got is essentially 
\begin{equation} \label{eq:xyz}
\x_i = \R_i \X = \R_i \D \Gama = \R_i \D \S_i^T \S_i \Gama.
\end{equation}
\begin{defn}
Consider a global sparse vector $\Gama$. Define $\gama_i = \S_i \Gama$ as its $i^{th}$ stripe representation.
\end{defn}
\noindent
Note that a stripe $\gama_i$ can be also seen as a group of $2n-1$ adjacent local sparse vectors $\alfa_j$ of length $m$ from $\Gama$, centered at location $\alfa_i$.
\begin{defn}
Consider a convolutional dictionary $\D$ defined by a local dictionary $\D_L$ of size $n \times m$. Define the stripe dictionary $\O$ of size $n \times (2n - 1)m$, as the one obtained by extracting $n$ consecutive rows from $\D$, followed by the removal of its zero columns, namely $\O = \R_i \D \S_i^T$.
\end{defn}
\noindent
Observe that $\O$, depicted in Figure \ref{PartialStripe}, is independent of $i$, being the same for all locations due to the union-of-Circulant-matrices structure of $\D$. In other words, the shift invariant property is satisfied for this model -- all patches share the same stripe dictionary in their construction. Armed with the above two definitions, Equation \eqref{eq:xyz} reads $\x_i = \O \gama_i$. 


From a different perspective, one can synthesize the signal $\X$ by a different interpretation of the relation $\X=\D\Gama$, shown in Figure \ref{GlobalDiagram}. The matrix $\D$ is a concatenation of $N$ vertical stripes of size $N\times m$, where each can be represented as $\R_i^T \D_L$. In other words, the vertical stripe is constructed by taking the small and local dictionary $\D_L$ and positioning it in the $i^{th}$ row. As we have already said, the same partitioning applies to $\Gama$, leading to the $\alfa_i$ ingredients. Thus,  
\begin{equation}
\X = \sum_i \R_i^T \D_L \alfa_i.
\end{equation}
Since $\alfa_i$ play the role of local sparse vectors, $\D_L \alfa_i$ are reconstructed patches (which are not the same as $\x_i = \O \gama_i$), and the sum above proposes a patch averaging approach as practiced in several papers \cite{Aharon2006,Zoran2011,Sulam2015}. This formulation provides another local interpretation of the convolutional model.

Yet a third interpretation of the very same signal construction can be suggested, in which the signal is seen as resulting from a sum of local/small atoms which appear in a small number of locations throughout the signal. This can be formally expressed as 
\begin{equation}
\X = \sum_{i=1}^m \d_i \ast \z_i, 
\end{equation}
where the vectors $\z_i\in \mathbb{R}^N$ are sparse maps encoding the location and coefficients of the $i^{th}$ atom \cite{Grosse2007}. In our context, $\Gama$ is simply the interlaced concatenation of all $\z_i$.

This model (adopting the last, convolutional, interpretation) has received growing attention in recent years in various applications. In \cite{Morup2008} a convolutional sparse coding framework was used for pattern detection in images and the analysis of instruments in music signals, while in \cite{Zhu2015} it was used for the reconstruction of 3D trajectories. The problem of learning the local dictionary $\D_L$ was also studied in several works \cite{Zeiler2010,Kavukcuoglu2010,Bristow2014,Huang2015}.

Different methods have been proposed for solving the convolutional sparse coding problem under an $\ell_1$-norm penalty. Commonly, these methods rely on the ADMM algorithm \cite{Boyd2011}, exploiting multiplications of vectors by the global dictionary in the Fourier domain in order to reduce the computational cost involved. The reader is referred to \cite{Bristow2014} for a thorough review of related methods. 
In essence, these are attempts to minimize a cost function which is a BP problem under the convolutional structure.
As a result, the theoretical results in our work will also apply to these methods, providing guarantees for the recovery of the underlying sparse vectors.
An interesting exception to the $\ell_1$ norm is the work reported in \cite{Szlam2010}, where the authors suggest an $\ell_0$ constraint on the sparse vectors. This algorithm, called Convolutional Matching Pursuit, was used to extract features from natural images. Up to the orthogonal projection step, this greedy method is a global OMP, for which we will also provide novel recovery guarantees. 
\begin{figure}[t]
	\centering
	\includegraphics[trim = 0 50 0 20,width=0.5\textwidth]{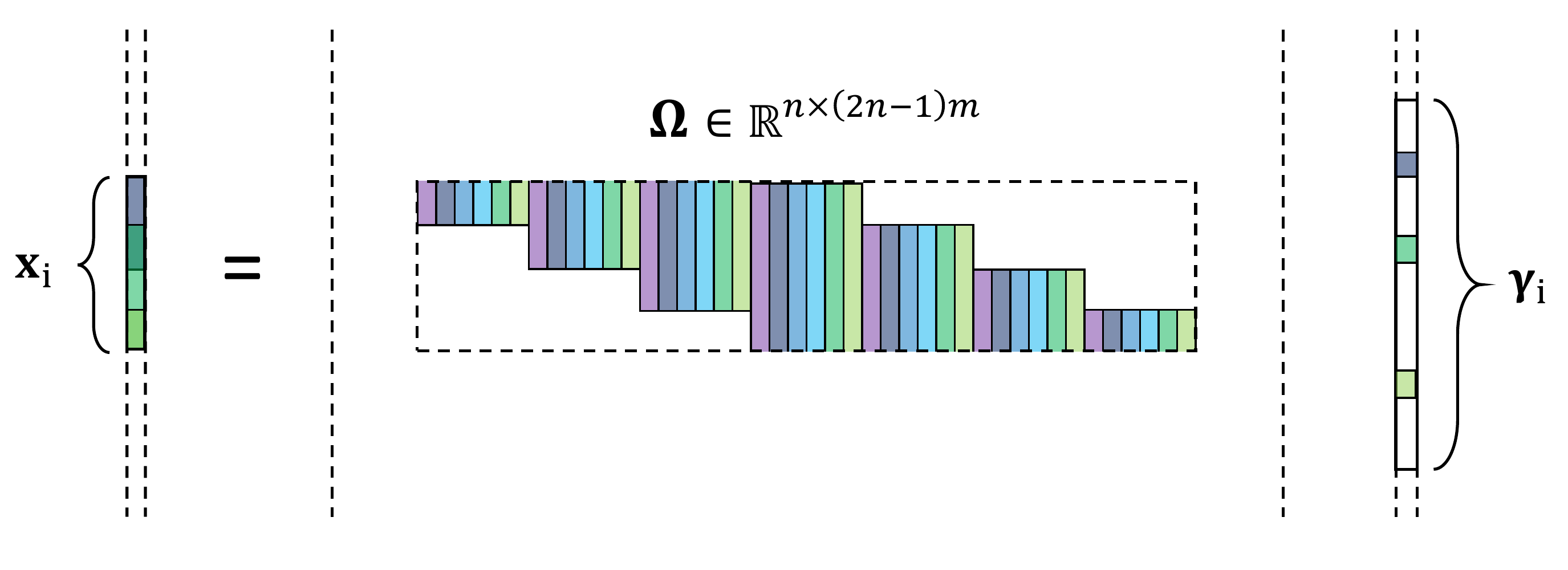}
	\caption{Stripe Dictionary}
	\label{PartialStripe}
	\vspace{-0.3cm}
\end{figure}
\section{From Global to Local Analysis}
\label{Sec:Global2Local}

Consider a sparse vector $\Gama$ of size $mN$ which represents a global (convolutional) signal. Assume further that this vector has a few $k \ll N$ non-zeros. If these were to be clustered together in a given stripe $\gama_i$, the local patch corresponding to this stripe would be very complex, and pursuit methods would likely fail in recovering it. On the contrary, consider the case where these $k$ non-zeros are spread all throughout the vector $\Gama$. This would clearly imply much simpler local patches, facilitating their successful recovery. This simple example comes to show the futility of the traditional global $\ell_0$-norm in the convolutional setting, and it will be the pillar of our intuition throughout our work.

\subsection{The $\ell_{0,\infty}$ Norm and the $P_{0,\infty}$ Problem}

Let us now introduce a measure that will provide a local notion of sparsity within a global sparse vector.
\begin{defn}
Define the $\Loi$ pseudo-norm of a global sparse vector $\Gama$ as
\begin{equation}
\|\Gama\|_{0,\infty} = \max_i \|\gama_i\|_0.
\end{equation}
\end{defn}
\noindent
In words, this quantifies the number of non-zeros in the densest stripe $\gamma_i$ of the global $\Gama$. This is equivalent to extracting all stripes from the global sparse vector $\Gama$, arranging them column-wise into a matrix $\A$ and applying the usual $\|\A\|_{0,\infty}$ norm -- thus, the name. Note that by constraining the $\Loi$ norm to be low, we are essentially limiting the sparsity of all the stripes $\gama_i$. Similar to $\ell_0$, in the $\Loi$ norm the non-negativity and triangle inequality properties hold, while homogeneity does not. Since non-negativity is trivial, we only prove the triangle inequality in Appendix \ref{sect:TraingIneqLoi}.

Armed with the above definition, we move now to define the $\Poi$ problem:
\begin{equation}
(\Poi): \quad \min_\Gama \quad \|\Gama\|_{0,\infty} \ \text{ s.t. }\ \D\Gama=\X.
\end{equation}
When dealing with a global signal, instead of solving the $P_0$ problem (defined in Equation \eqref{Eq:P0problem}) as is commonly done, we aim to solve the above defined objective instead. The key difference is that we are not limiting the overall number of zeros in $\Gama$, but rather putting a restriction on its local density. 

\subsection{Global versus Local Bounds}
\label{GlobalVsLocal}
As mentioned previously, theoretical bounds are often given in terms of the mutual coherence of the dictionary. In this respect, a lower bound on this value is much desired. In the case of our convolution sparse model, this value quantifies not only the correlation between the atoms in $\D_L$, but also the correlation between their shifts. Though in a different context, a bound for this value was derived in \cite{Welch1974}, and it is given by
\begin{equation} \label{eq:mu_bound}
\mu(\D)\geq\sqrt{\frac{m-1}{m(2n-1)-1}}.
\end{equation}
For example, if $m=1$ (one local atom with all its shifts), this suggests that $\D$ might be an orthogonal matrix, and thus $\mu(\D)=0$. Going to the other extreme, for a large value of $m$ one obtains that the best possible coherence is \mbox{$\mu(\D)\approx \frac{1}{\sqrt{2n}}$} -- this is a very high value (e.g., if $n=128$, this coherence bound is $1/16$), considering the fact that it characterizes the whole global dictionary. This implies that if we are to apply BP or OMP to recover the sparsest $\Gama$ that represents $\X$, the classical sparse approximation results \cite{Bruckstein2009} would allow merely $O(\sqrt{n})$ non-zeros in \textbf{all} $\Gama$, for any $N$, no matter how long $\X$ is!

As we shall see next, the situation is not as grave as may seem, due to our migration from $P_0$ to $\Poi$.
Leveraging on the definitions from the previous subsection, we will provide recovery guarantees that will have a local flavor, and the bounds will be given in terms of the number of non-zeros in the densest stripe. This way, we will show that the guarantee conditions can be significantly enhanced to $O(\sqrt{n})$ non-zeros \emph{locally} rather than \emph{globally}.

\vspace{-0.2cm}
\section{Theoretical Study} \label{Theoretical Study}
\label{Sec:TheoStudy}
As motivated in the previous section, the concerns of uniqueness, recovery guarantees and stability of sparse solutions in the convolutional case require special attention. We now formally address these questions by closely following the path taken in \cite{Elad_Book}, carefully generalizing each and every statement to the global-local model discussed here.

Before proceeding onto theoretical grounds, we briefly summarize, for the convenience of the reader, all notations used throughout this work in Table \ref{Table:Notations}. Note the unorthodox choice of capital letters for global vectors and lowercase for local ones.

\vspace{0.5cm}

\begin{table}[t] \centering
\begin{tabular}{|p{1.85cm}@{:\quad}l|} \hline
$N$ & length of the global signal. \\ \hline
$n$ & size of a local atom or a local signal patch. \\ \hline
$m$ & \pbox{20cm}{number of unique local atoms (filters) or the number\\ of Circulant matrices.} \\[0.2cm] \hline
$\X$, $\Y$ and $\E$ & \pbox{20cm}{global signals of length $N$, where generally\\ $\Y=\X+\E$.} \\ \hline
$\D$ & global dictionary of size $N\times mN$. \\ \hline
$\Gama$ and $\Delt$ & global sparse vectors of length $mN$. \\ \hline
$\Gamma_i$ and $\Delta_i$ & the $i^{th}$ entry in $\Gama$ and $\Delt$, respectively. \\ \hline
$\D_L$ & local dictionary of size $n\times m$. \\ \hline
$\O$ & \pbox{20cm}{stripe dictionary, of size $n\times(2n-1)m$, which \\ contains all possible shifts of $\D_L$.} \\ \hline
$\alfa_i$ & local sparse code of size $m$. \\ \hline
$\gama_i$ and $\delt_i$ & \pbox{20cm}{a stripe of length $(2n-1)m$ extracted \\ from the global vectors $\Gama$ and $\Delt$, respectively.} \\ \hline
$\gama_{i,s}$ and $\delt_{i,s}$ & \pbox{20cm}{a local sparse vector of length $m$ which corresponds \\ to the $s^{th}$ portion inside $\gama_i$ and $\delt_i$, respectively.} \\ \hline
\end{tabular}
\caption{Summary of notations used throughout the paper.}
\label{Table:Notations}
\vspace{-0.4cm}
\end{table}

\vspace{-0.6cm}
\subsection{Uniqueness and Stripe-Spark}

Just as it was initially done in the general sparse model, one might ponder about the uniqueness of the sparsest representation in terms of the $\Loi$ norm. More precisely, does a unique solution to the $\Poi$ problem exist? and under which circumstances? In order to answer these questions we shall first extend our mathematical tools, in particular the characterization of the dictionary, to the convolutional scenario.

In Section \ref{Sec:Preliminaries} we recalled the definition of the Spark of a general dictionary $\D$. In the same spirit, we can propose the following:
\begin{defn}
Define the Stripe-Spark of a convolutional dictionary\ $\D$ as
\begin{equation}
\sigma_\infty(\D)= \min_\Delt \quad \|\Delt\|_{0,\infty} \ \text{ s.t. }\  \Delt\neq 0,\ \D\Delt=0.
\end{equation}
\end{defn}
\noindent
In words, the Stripe-Spark is defined by the sparsest non-zero vector, in terms of the $\Loi$ norm, in the null space of $\D$. 
Next, we shall use this definition in order to formulate an uncertainty and a uniqueness principle for the $\Poi$ problem that emerges from it.
\begin{thm}{(Uncertainty and uniqueness using Stripe-Spark):}
	Let $\D$ be a convolutional dictionary. If a solution $\Gama$ obeys \mbox{$\|\Gama\|_{0,\infty}<\frac{1}{2}\sigma_\infty$}, then this is necessarily the global optimum for the $\Poi$ problem for the signal $\D\Gama$.
\end{thm}

\begin{proof}
	Let $\hat{\Gama}\neq\Gama$ be an alternative solution. Then $\D\left(\Gama-\hat{\Gama}\right)=0$. By definition of the Stripe-Spark
	\begin{equation}
	\|\Gama-\hat{\Gama}\|_{0,\infty}\geq\sigma_\infty.
	\end{equation}
	Using the triangle inequality of the $\Loi$ norm,
	\begin{equation}
	\|\Gama\|_{0,\infty}+\|\hat{\Gama}\|_{0,\infty}\geq\|\Gama-\hat{\Gama}\|_{0,\infty}\geq\sigma_\infty.
	\end{equation}
	This result poses an uncertainty principle for $\Loi$ sparse solutions of the system $\X = \D\Gama$, suggesting that if a very sparse solution is found, all alternative solutions must be much denser. Since $\|\Gama\|_{0,\infty}<\frac{1}{2}\cdot\sigma_\infty$, we must have that $\|\hat{\Gama}\|_{0,\infty}>\frac{1}{2}\cdot\sigma_\infty$, or in other words, every solution other than $\Gama$ has higher $\Loi$ norm, thus making $\Gama$ the global solution for the $\Poi$ problem.
\end{proof}

\subsection{Lower Bounding the Stripe-Spark}

In general, and similar to the Spark, calculating the Stripe-Spark is computationally intractable. Nevertheless, one can bound its value using the global mutual coherence defined in Section \ref{Sec:Preliminaries}. Before presenting such bound, we formulate and prove a Lemma that will aid our analysis throughout this paper.

\begin{lemma}{}
\label{lemma:GershgorinConvolutional}
Consider a convolutional dictionary $\D$, with mutual coherence $\mu(\D)$, and a support $\mathcal{T}$ with $\Loi$ norm\footnote{Note that specifying the $\Loi$ of a support rather than a sparse vector is a slight abuse of notation, that we will nevertheless use for the sake of simplicity.} equal to $k$. Let $\G^\mathcal{T} = \D_\mathcal{T}^T \D_\mathcal{T}$, where $\D_\mathcal{T}$ is the matrix $\D$ restricted to the columns indicated by the support $\mathcal{T}$. Then, the eigenvalues of this Gram matrix, given by $\lambda_i \left(\G^\mathcal{T}\right)$, are bounded by
\begin{equation}
1-(k-1)\mu(\D) \ \leq \ \lambda_i \left(\G^\mathcal{T}\right) \ \leq \ 1+(k-1)\mu(\D).
\end{equation}
\end{lemma}

\begin{proof}
	From Gerschgorin's theorem, the eigenvalues of the Gram matrix $\G^\mathcal{T}$ reside in the union of its Gerschgorin circles. The $j^{th}$ circle, corresponding to the $j^{th}$ row of $\G^\mathcal{T}$, is centered at the point $\G^{\mathcal{T}}_{j,j}$ (belonging to the Gram's diagonal) and its radius equals the sum of the absolute values of the off-diagonal entries; i.e., $\sum_{i,i\neq j} |\G^{\mathcal{T}}_{j,i}|$. Notice that both indices $i,j$ correspond to atoms in the support $\mathcal{T}$. Because the atoms are normalized, $\forall\ j,\ \G^{\mathcal{T}}_{j,j} = 1$, implying that all Gershgorin disks are centered at $1$. Therefore, all eigenvalues reside inside the circle with the largest radius. Formally,
	\begin{equation} \label{eq:max_radius}
	\big| \lambda_i\left(\G^\mathcal{T}\right) - 1 \big| \leq \max_j \sum_{i,i\neq j}  \big| \G^{\mathcal{T}}_{j,i} \big| = \max_j \sum_{\substack{i,i\neq j \\ i,j \in \mathcal{T}}} | \d_j^T\d_i \big|.
	\end{equation}
	On the one hand, from the definition of the mutual coherence, the inner product between atoms that are close enough to overlap is bounded by $\mu(\D)$. On the other hand, the product $\d_j^T\d_i$ is zero for atoms $\d_i$ too far from $\d_j$ (i.e., out of the stripe centered at the $j^{th}$ atom). Therefore, we obtain:
	\begin{equation}
	\sum_{\substack{i,i\neq j \\ i,j \in \mathcal{T}}} | \d_j^T\d_i | \leq (k - 1)\ \mu(\D),
	\end{equation}
	where $k$ is the maximal number of non-zero elements in a stripe, defined previously as the $\Loi$ norm of $\mathcal{T}$. Note that we have subtracted $1$ from $k$ because we must omit the entry on the diagonal. Putting this back in Equation \eqref{eq:max_radius}, we obtain
	\begin{equation}
	\big| \lambda_i\left(\G^\mathcal{T}\right) - 1 \big| \leq \max_j \sum_{\substack{i,i\neq j \\ i,j \in \mathcal{T}}} | \d_j^T\d_i \big| \leq\ (k - 1)\ \mu(\D).
	\end{equation}
	From this we obtain the desired claim.
\end{proof}


Based on this, we dive into the next theorem.

\begin{thm}{(Lower bounding the Stripe-Spark via the local coherence):}
	For a convolutional dictionary $\D$ with mutual coherence $\mu(\D)$, the Stripe-Spark can be lower-bounded by
	\begin{equation}
	\sigma_\infty(\D)\geq 1+\frac{1}{\mu(\D)}.
	\end{equation}
\end{thm}

\begin{proof}
	Let $\Delt$ be a vector such that $\Delt\neq\mathbf{0}$ and $\D\Delt=\mathbf{0}$.
	Note that we can write
	\begin{equation} \label{eq:relation}
	\D_\mathcal{T} \Delt_\mathcal{T} = \mathbf{0},
	\end{equation}
	where $\Delt_\mathcal{T}$ is the vector $\Delt$ restricted to its support $\mathcal{T}$, and $\D_\mathcal{T}$ is the dictionary composed of the corresponding atoms.
	Consider now the Gram matrix, $\G^{\mathcal{T}} = \D_\mathcal{T}^T \D_\mathcal{T}$, which corresponds to a portion extracted from the global Gram matrix $\D^T\D$. The relation in Equation \eqref{eq:relation} suggests that $\D_\mathcal{T}$ has a nullspace, which implies that its Gram matrix must have at least one eigenvalue equal to zero. Using Lemma \ref{lemma:GershgorinConvolutional}, the lower bound on the eigenvalues of $\G^{\mathcal{T}}$ is given by $1-(k-1)\mu(\D)$, where $k$ is the $\Loi$ norm of $\Delt$. Therefore, we must have that $1-(k-1)\mu(\D) \leq 0$, or equally $k \geq 1 + \frac{1}{\mu(\D)}$. We conclude that a vector $\Delt$, which is in the null-space of $\D$, must always have an $\Loi$ norm of at least $1+\frac{1}{\mu(\D)}$, and so the Stripe-Spark $\sigma_\infty$ is also bounded by this number.
\end{proof}

Using the above derived bound and the uniqueness based on the Stripe-Spark we can now formulate the following theorem:

\begin{thm}{(Uniqueness using mutual coherence):}
	Let $\D$ be a convolutional dictionary with mutual coherence $\mu(\D)$. If a solution $\Gama$ obeys $\|\Gama\|_{0,\infty}<\frac{1}{2}(1+\frac{1}{\mu(\D)})$, then this is necessarily the sparsest (in terms of $\Loi$ norm) solution to $\Poi$ with the signal $\D\Gama$.
\end{thm}

At the end of Section \ref{Sec:Global2Local} we mentioned that for $m\gg 1$, the classical analysis would allow an order of $O(\sqrt{n})$ non-zeros all over the vector $\Gama$, regardless of the length of the signal $N$. In light of the above theorem, in the convolutional case, the very same quantity of non-zeros is allowed locally per stripe, implying that the overall number of non-zeros in $\Gama$ grows linearly with $N$.

\subsection{Recovery Guarantees for Pursuit Methods}
\label{Sec:BPnoiseless}

In this subsection, we attempt to solve the $\Poi$ problem by employing two common pursuit methods: the Orthogonal Matching Pursuit (OMP) and the Basis Pursuit (BP). Leaving aside the computational burdens of running such algorithms, which will be addressed in the second part of this work, we now consider the theoretical aspects of their success. Observe that in the coming discussion we use these two algorithms in their natural form, being oblivious to the $\Loi$ objective they are serving. Further work is required to develop OMP and BP versions that are aware of this specific goal, and thus may benefit from it.  

Previous works \cite{Donoho2003,Tropp2004} have shown that both OMP and BP succeed in finding the sparsest solution to the $P_0$ problem if the cardinality of the representation is known a priori to be lower than $\frac{1}{2}(1+\frac{1}{\mu(\D)})$. That is, we are guaranteed to recover the underlying solution as long as the \emph{global sparsity} is less than a certain threshold. In light of the discussion in Section \ref{GlobalVsLocal}, these values are pessimistic in the convolutional setting. By migrating from $P_0$ to the $\Poi$ problem, we show next that both algorithms are in fact capable of recovering the underlying solutions under far weaker assumptions.

\begin{thm}{(Global OMP recovery guarantee using $\Loi$ norm):} \label{thm:OMPSuccess}
	Given the system of linear equations $\X = \D\Gama$, if a solution $\Gama$ exists satisfying
	\begin{equation} \label{eq:Loi_condition}
	\|\Gama\|_{0,\infty} < \ \frac{1}{2}\left(1+\frac{1}{\mu(\D)}\right),
	\end{equation}
	then OMP is guaranteed to recover it.
\end{thm}
Note that if we assume $\|\Gama\|_{0,\infty} < \ \frac{1}{2}\left(1+\frac{1}{\mu(\D)}\right)$, according to our uniqueness theorem, the solution obtained by the OMP is the unique solution to the $\Poi$ problem and thus OMP finds the solution with the minimal $\Loi$ norm. Next, we claim that under the same conditions the BP algorithm is guaranteed to succeed as well. The proofs of these two theorems are presented in Appendix \ref{Sec:ProofsTheoremsNoiseless}.
\begin{thm}{(Global Basis Pursuit recovery guarantee using the $\Loi$ norm):} \label{thm:BPSuccess}
	For the system of linear equations $\D\Gama=\X$, if a solution $\Gama$ exists obeying
	\begin{equation}
		\|\Gama\|_{0,\infty}<\frac{1}{2}\left(1+\frac{1}{\mu(\D)}\right),
	\end{equation}
	then Basis Pursuit is guaranteed to recover it.
\end{thm}

Before moving on, we would like to highlight again the implications of the aforementioned claims. The recovery guarantees for both pursuit methods have now become \emph{independent of the global signal dimension and sparsity}. Instead, the condition for success is given in terms of the \emph{local} concentration of non-zeros of the global sparse vector. Moreover, the number of non-zeros permitted per stripe under the current bounds is in fact the same number previously allowed globally.

\subsection{Experiments}
In this subsection we intend to provide numerical results that corroborate the above presented theoretical bounds. While doing so, we will shed light on the performance of the OMP and BP algorithms in practice, as compared to our previous analysis.

In \cite{Soltanalian2014} an algorithm was proposed to construct a local dictionary such that all its aperiodic auto-correlations and cross-correlations are low. This, in our context, means that the algorithm attempts to minimize the mutual coherence of the dictionary $\D_L$ and all of its shifts, decreasing the global mutual coherence as a result. We use this algorithm to numerically build a dictionary consisting of two atoms ($m=2$) with patch size $n=64$. The theoretical lower bound on the $\mu(\D)$ presented in Equation \eqref{eq:mu_bound} under this setting is approximately $0.063$, and we manage to obtain a mutual coherence of $0.09$ using the aforementioned method. With these atoms we construct a convolutional dictionary with global atoms of length $N = 640$.

Once the dictionary is fixed, we generate sparse vectors with random supports of (global) cardinalities in the range $\left[1,300\right]$. The non-zero entries are drawn from random independent and identically-distributed Gaussians with mean equal to zero and variance equal to one. Given these sparse vectors, we compute their corresponding global signals and attempt to recover them using the global OMP and BP. We perform $500$ experiments per each cardinality and present the probability of success as a function of the representation's $\Loi$ norm. We define the success of the algorithm as the full recovery of the true sparse vector. The results for the experiment are presented in Figure \ref{fig:phase_transition_both}. The theorems provided in the previous subsection guarantee the success of both OMP and BP as long as the $\|\Gama \|_{0,\infty} \leq 6$, as  $\frac{1}{2}\left(1+\frac{1}{\mu(\D)}\right) \approx 6$. 
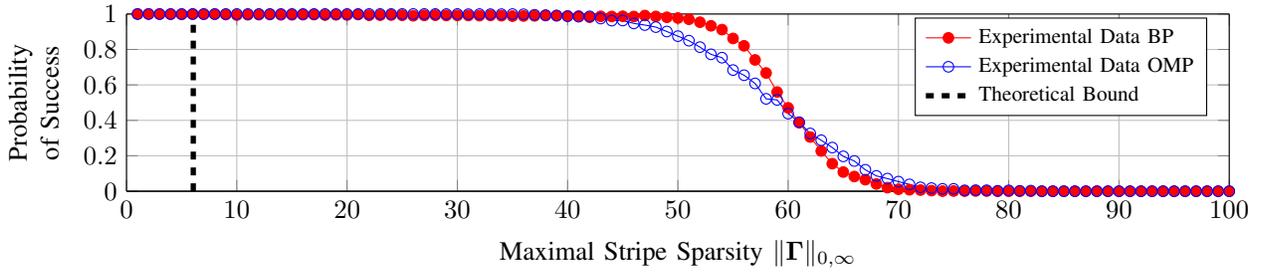
\begin{figure*}[t]
	\vspace{-.5cm}
	\centering
	\setlength\figureheight{0.13\textwidth}
	\setlength\figurewidth{0.85\textwidth}
	\input{phase_transition_both.tikz}
	\caption{Probability of success of OMP and BP at recovering the true convolutional sparse code. The theoretical guarantee is presented on the same graph.}
	\label{fig:phase_transition_both}
	\vspace{-0.3cm}
\end{figure*}

As can be seen from these results, the theoretical bound is far from being tight. However, in the traditional sparse representation model the corresponding bounds have the same loose flavor \cite{Bruckstein2009}. This kind of results is in fact expected when using such a worst-case analysis. Tighter bounds could likely be obtained by a probabilistic study, which we leave for future work.


\section{Shifted Mutual Coherence and Stripe Coherence}
\label{Sec:ShiftedMutualCoherence}
When considering the mutual coherence $\mu(\D)$, one needs to look at the maximal correlation between every pair of atoms in the global dictionary. One should note, however, that atoms having a non-zero correlation must have overlapping supports. As we see, $\mu(\D)$ provides a bound for these values independently of the amount of overlap. 
One could go beyond this characterization of the convolutional dictionary by a single value and propose to bound all the inner products between atoms for a \emph{given shift}. In this section we briefly explore this direction of analysis, introducing new tools for this model and addressing the theoretical consequences that they convey. We will only present the main points of these results here for the sake of brevity; the interested reader can find a more detailed discussion on this matter in Appendix \ref{App:MoreOnShiftedMu}.

Recall that $\O$ is defined as a stripe extracted from the global dictionary $\D$, as explained in Section \ref{Sec:Conv_Model}. Consider the sub-system given by $\x_i=\O\gama_i$, corresponding to the $i^{th}$ patch in $\X$. Note that $\O$ can be split into a set of $2n-1$ blocks of size $n\times m$, where each block is denoted by $\O_s$, i.e., 
\begin{equation}
\O=[\O_{-n+1},\dots,\O_{-1},\O_0,\O_1,\dots,\O_{n-1}],
\end{equation}
as shown previously in Figure \ref{PartialStripe}. Similarly, $\gama_i$ can be split into a set of $2n-1$ vectors of length $m$, each denoted by $\gama_{i,s}$ and corresponding to $\O_s$. In other words, $\gama_i=[\gama_{i,-n+1}^T,\dots,\gama_{i,-1}^T,\gama_{i,0}^T,\gama_{i,1}^T,\dots,\gama_{i,n-1}^T]^T$. Note that previously we denoted local sparse vectors of length $m$ by $\alfa_j$. Yet, we will also denote them by $\gama_{i,s}$ in order to emphasize the fact that they correspond to the $s^{th}$ shift inside $\gama_i$. Denote the number of non-zeros in $\gama_i$ as $n_i$. We can also write $n_i=\displaystyle\sum_{s=-n+1}^{n-1}n_{i,s}$,
where $n_{i,s}$ is the number of non-zeros in each $\gama_{i,s}$.

\begin{defn}
Define the shifted mutual coherence $\mu_s$ by
\begin{equation}
\mu_s = \underset{i,j}{\max} \quad |\langle \d^0_i,\d^s_j  \rangle|,
\end{equation}
where $\d^0_i$ is a column extracted from $\O_0$, $\d^s_j$ is extracted from $\O_s$, and we require\footnote{The condition $i\neq j$ if $s=0$ is necessary so as to avoid the inner product of an atom by itself.} that $i\neq j$ if $s=0$.
\end{defn}

The above definition can be seen as a generalization of the mutual coherence for the shift-invariant local model presented in Section \ref{Sec:Conv_Model}. Indeed, $\mu_s$ characterizes $\O$ just as $\mu(\D)$ characterizes the coherence of a general dictionary. Note that if $s=0$ the above definition boils down to the traditional mutual coherence of $\D_L$, i.e., $\mu_0=\mu(\D_L)$. It is important to stress that the atoms used in the above definition {\em are normalized globally} according to $\D$ and not $\O$. In Appendix \ref{App:MoreOnShiftedMu} we comment on several interesting properties of this measure.

Following the definition of the shifted mutual coherence, for a given stripe $i$ from the linear system of equations $\X=\D\Gama$, we can formulate a new measure: 

\begin{defn}
The stripe coherence is defined as
\begin{equation}
\zeta(\gama_i) = \sum_{s=-n+1}^{n-1} n_{i,s}\ \mu_s.
\end{equation}
\end{defn}
\noindent
According to this definition, each stripe has a coherence given by the sum of its non-zeros weighted by the shifted mutual coherence. As a particular case, if all $k$ non-zeros correspond to atoms in the center sub-dictionary, $\D_L$, this becomes $\mu_0k$. Note that unlike the traditional mutual coherence, this new measure depends on the location of the non-zeros in $\Gama$ -- it is a function of the support of the sparse vector, and not just of the dictionary. As such, it characterizes the correlation between the atoms participating in a given stripe. In what follows, we will use the notation $\zeta_i$ for $\zeta(\gama_i)$.

Next, we will present results based on these measures. Although these theorems are generally sharper, they are harder to grasp. We begin with a recovery guarantee for the OMP and BP algorithms, followed by a discussion on their implications.

\begin{thm}{(Global OMP recovery guarantee using the stripe coherence):}
	\label{thm:OMPSuccess_stripeCoherence}
	Given the system of linear equations $\X = \D\Gama$, if a solution $\Gama$ exists satisfying \begin{equation} \label{eq:stripe_coherence_condition3}
	\max_i\ \zeta_i = \max_i\sum_{s=-n+1}^{n-1} n_{i,s} \mu_s < \ \frac{1}{2}\left(1+\mu_0\right),
	\end{equation}
	then OMP is guaranteed to recover it.
\end{thm}

\begin{thm}{(Global BP recovery guarantee using the stripe coherence):}
	Given the system of linear equations $\X = \D\Gama$, if a solution $\Gama$ exists satisfying \begin{equation}
	\max_i\ \zeta_i = \max_i\sum_{s=-n+1}^{n-1} n_{i,s} \mu_s < \ \frac{1}{2}\left(1+\mu_0\right),
	\end{equation}
	then Basis Pursuit is guaranteed to recover it.
\end{thm}
\noindent The corresponding proofs are similar to their counterparts presented in the preceding section but require a more delicate analysis; one of them is thoroughly discussed in Appendix \ref{App:OMPproofStripeCoherence}. 

In order to provide an intuitive interpretation for these results, the above bounds can be tied to a concrete number of non-zeros per stripe. First, notice that requiring the maximal stripe coherence to be less than a certain threshold is equal to requiring the same for every stripe:
\begin{equation}
\forall i\quad\sum_{s=-n+1}^{n-1} n_{i,s} \mu_s< \ \frac{1}{2}\left(1+\mu_0\right).
\end{equation}
Multiplying and dividing the left-hand side of the above inequality by $n_i$ and rearranging the resulting expression, we obtain
\begin{equation}
\forall i \quad n_i< \ \frac{1}{2}\frac{1+\mu_0}{\sum_{s=-n+1}^{n-1} \frac{n_{i,s}}{n_i} \mu_s }.
\end{equation}
Define $\bar{\mu}_i=\sum_{s=-n+1}^{n-1}\frac{n_{i,s}}{n_i} \mu_s$. Recall that $\sum_{s=-n+1}^{n-1} \frac{n_{i,s}}{n_i}=1$ and as such $\bar{\mu}_i$ is simply the (weighted) average shifted mutual coherence in the $i^{th}$ stripe. Putting this definition into the above condition, the inequality becomes
\begin{equation}
\forall i\quad n_i< \frac{1}{2}\left(\frac{1}{\bar{\mu}_i}+\frac{\mu_0}{\bar{\mu}_i}\right).
\end{equation}
Thus, the condition in \eqref{eq:stripe_coherence_condition3} boils down to requiring the sparsity of all stripes to be less than a certain number. Naturally, this inequality resembles the one presented in the previous section for the OMP and BP guarantees. The reader might wonder about how they are related. In Appendix \ref{App:MoreOnShiftedMu} we prove that under the assumption that $\mu(\D)=\mu_0$, the shifted mutual coherence condition is at least as strong as the original one. 

As a final note, the shifted mutual coherence, $\mu_s$, is a considerably more informative measure than the standard mutual coherence. In some applications, the signals created by the convolutional dictionary are built of atoms which are known a priori to be separated by some minimal lag, or shift. In radio communications, for example, such a situation appears when there exists a minimal time between consecutive transmissions \cite{He2009}. In these cases, knowing how the correlation between the atoms depends on their shifts is fundamental for the design of the dictionary and its utilization.

\section{Conclusion and Future Work}
\label{Sec:Conclusions}

In the first part of this work we have presented a formal analysis of the convolutional sparse representation model. In doing so, we have reformulated the objective of the global pursuit, introducing the $\Loi$ norm and the corresponding $\Poi$ problem, and proven the uniqueness of its solution. By migrating from the $P_0$ to the $\Poi$ problem, we were able to provide meaningful guarantees for the success of popular algorithms in the noiseless case, improving on traditional bounds which were shown to be very pessimistic under the convolutional case. In order to achieve such results, we have generalized a series of concepts such as Spark and the mutual coherence to their counterparts in the convolutional setting. 

One of the cardinal motivations for this work was a series of recent practical methods addressing the convolutional sparse coding problem; and in particular, the need for their theoretical foundation. However, our results are as of yet not directly applicable to these, as we have restricted our analysis to the ideal case of noiseless signals. The natural extension to this work is therefore the study of signals under noise contamination and model imperfections. This is indeed the path we undertake in part II of our work, exploring the question of whether the convolutional model remains stable in the presence of noise. Moreover, we show how to decompose and solve the global pursuit by performing merely local operations. This will tie the algorithmic solutions for the convolutional model to patch-based methods, which are the current practice in state-of-the-art signal and image restoration.



\section{Acknowledgements}
The research leading to these results has received funding from the European Research Council under European Union’s Seventh Framework Programme, ERC Grant agreement no. 320649. The authors would like to thank Dmitry Batenkov, Yaniv Romano and Raja Giryes for the prolific conversations and most useful advice which helped shape this work.

\appendices
\renewcommand{\theequation}{A-\arabic{equation}}
\setcounter{equation}{0}  
\vspace{-0.2cm} 
\section{Triangle Inequality for the $\Loi$ Norm} 
\label{sect:TraingIneqLoi}
\begin{thm}
The triangle inequality holds for the $\Loi$ norm.
\end{thm}

\begin{proof}
Let $\Gama^1$ and $\Gama^2$ be two global sparse vectors. Denote the $i^{th}$ stripe extracted from each as $\gama^1_i$ and $\gama^2_i$, respectively. Notice that
\begin{align*}
\| \Gama^1+\Gama^2 \|_{0,\infty} =\max_i\|\gama^1_i+\gama^2_i\|_0 & \leq\max_i \left(\|\gama^1_i\|_0+\|\gama^2_i\|_0\right)\\
\leq \max_i\|\gama^1_i\|_0 +\max_i\|\gama^2_i\|_0 & =\|\Gama^1\|_{0,\infty}+\|\Gama^2\|_{0,\infty}.
\end{align*}
In the first inequality we have used the triangle inequality of the $\ell_0$ norm.
\end{proof}

\renewcommand{\theequation}{B-\arabic{equation}}
\setcounter{equation}{0}  
\section{Guarantees for Pursuit Methods for $\Poi$} \label{Sec:ProofsTheoremsNoiseless}
In this section we prove both theorems presented in Section \ref{Sec:TheoStudy}, which guarantee the success of OMP and BP in solving the $\Poi$ problem. We begin by presenting the OMP proof.

\subsection{OMP Success Guarantee (Proof of Theorem \ref{thm:OMPSuccess})}
\begin{proof}
	Denoting by $\mathcal{T}$ the support of the solution $\Gama$, we can write
	\begin{equation}
	\X = \D\Gama = \sum_{t\in \mathcal{T}} \Gamma_t \d_t.
	\label{GlobalExpression3}
	\end{equation}
	Suppose, without loss of generality, that the sparsest solution has its largest coefficient (in absolute value) in $\Gamma_i$. 
	For the first step of the OMP to choose one of the atoms in the support, we require
	\begin{equation}
	|\d_i^T \X | > \max_{j\notin\mathcal{T}} | \d_j^T \X |.
	\end{equation}
	Substituting Equation \eqref{GlobalExpression3} in this requirement we obtain
	\begin{equation} \label{eq:inequality3}
	\left| \sum_{t\in \mathcal{T}}\Gamma_t\d_t^T\d_i \right| > \max_{j\notin\mathcal{T}} \left| \sum_{t\in \mathcal{T}}\Gamma_t\d_t^T\d_j \right|.
	\end{equation}
	Using the reverse triangle inequality, the assumption that the atoms are normalized, and that $|\Gamma_i|\geq|\Gamma_t|$, we construct a lower bound for the left hand side:
	\begin{align}
	\left| \sum_{t\in \mathcal{T}}\Gamma_t\d_t^T\d_i \right|
	\geq& |\Gamma_i| - \sum_{t \in \mathcal{T},t\neq i}|\Gamma_t|\cdot|\d_t^T\d_i | \\
	\geq& |\Gamma_i| - |\Gamma_i|\sum_{t \in \mathcal{T},t\neq i}|\d_t^T\d_i |.
	\end{align}
	Consider the stripe which completely contains the $i^{th}$ atom as shown in Figure \ref{Fig:details}. Notice that $\d_t^T\d_i$ is zero for every atom too far from $\d_i$ because the atoms do not overlap. Denoting the stripe which fully contains the $i^{th}$ atom as $p(i)$ and its support as $\mathcal{T}_{p(i)}$, we can restrict the summation as:
	\begin{equation} \label{Eq:EquationFromOMPproof}
	\left| \sum_{t\in \mathcal{T}}\Gamma_t\d_t^T\d_i \right| \geq |\Gamma_i| - |\Gamma_i|\sum_{t \in \mathcal{T}_{p(i)},t\neq i}|\d_t^T\d_i |.
	\end{equation}
	We can bound the right side by using the number of non-zeros in the support $\mathcal{T}_{p(i)}$, denoted by $n_{p(i)}$, together with the definition of the mutual coherence, obtaining:
	\begin{figure} 
		\includegraphics[trim = -80 80 80 0, width=.4\textwidth]{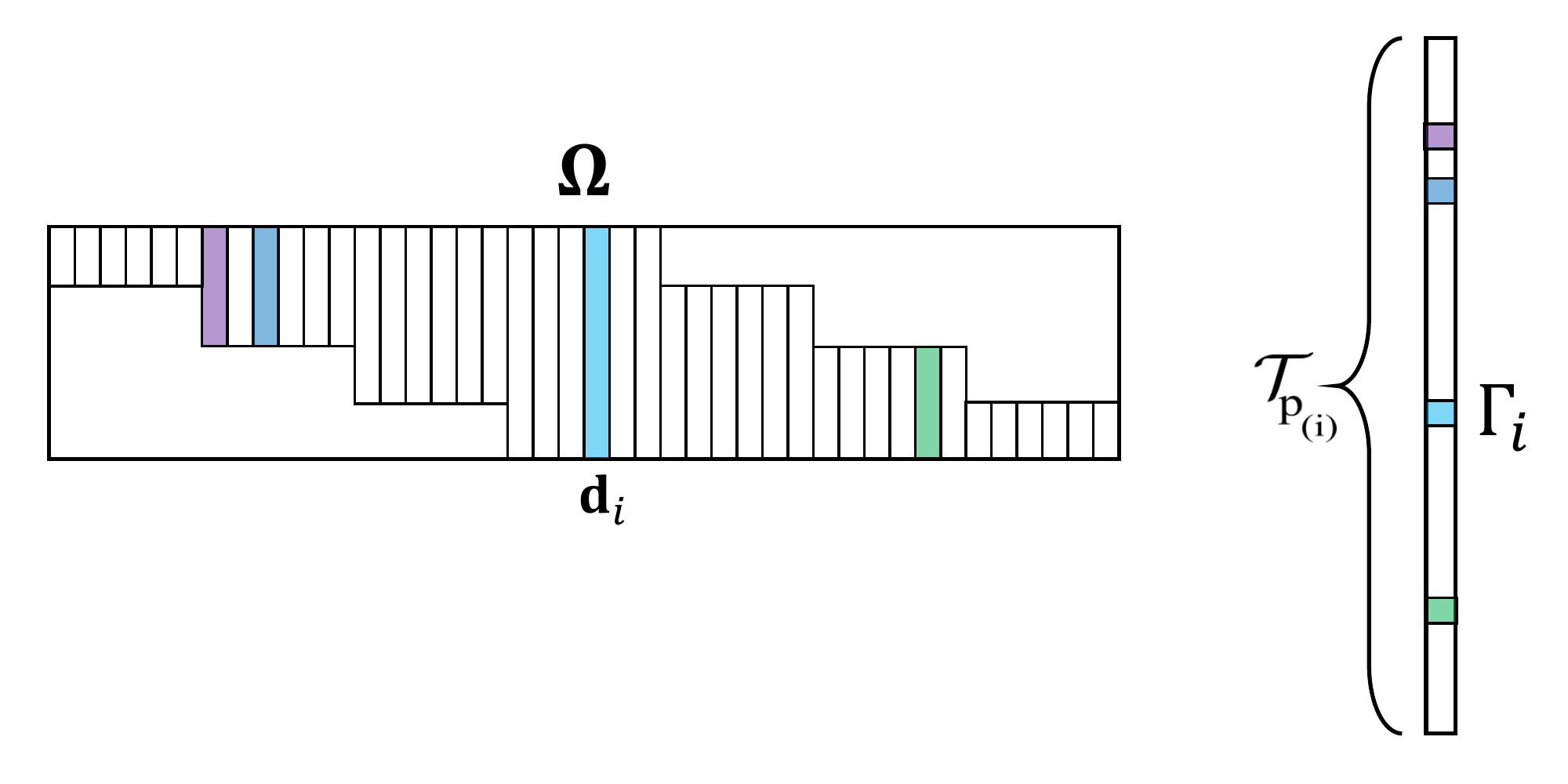}
		\caption{The $p_{(i)}$ stripe of atom $\d_i$.}
		\label{Fig:details}
		\vspace{-0.3cm}
	\end{figure}
	\begin{equation}
	\left| \sum_{t\in \mathcal{T}}\Gamma_t\d_t^T\d_i \right| \geq |\Gamma_i| -  |\Gamma_i| \cdot (n_{p(i)}-1)\cdot\mu(\D).
	\end{equation}
	Using the definition of the $\Loi$ norm, we obtain
	\begin{equation}
	\left| \sum_{t\in \mathcal{T}}\Gamma_t\d_t^T\d_i \right| \geq |\Gamma_i| -  |\Gamma_i| \cdot (\|\Gama\|_{0,\infty}-1)\cdot\mu(\D).
	\end{equation}
	Now, we construct an upper bound for the right hand side of Equation \eqref{eq:inequality3}, using the triangle inequality and the fact that $|\Gamma_i|$ is the maximal value in the sparse vector:
	\begin{align} \label{eq:SkipSteps2}
	\max_{j\notin\mathcal{T}}\left| \sum_{t\in \mathcal{T}}\Gamma_t\d_t^T\d_j \right| &\leq \max_{j\notin\mathcal{T}}\sum_{t \in \mathcal{T}} |\Gamma_t|\cdot|\d_t^T\d_j |\\
	&\leq |\Gamma_i|\max_{j\notin\mathcal{T}}\sum_{t \in \mathcal{T}}|\d_t^T\d_j |.
	\end{align}
	Relying on the same rational as above, we obtain:
	\begin{align}
	&& \max_{j\notin\mathcal{T}}\left| \sum_{t\in \mathcal{T}}\Gamma_t\d_t^T\d_j \right| & \leq |\Gamma_i|\max_{j\notin\mathcal{T}}\sum_{t \in \mathcal{T}_{p(j)}}|\d_t^T\d_j |\\
	&&\leq |\Gamma_i|\max_{j\notin\mathcal{T}}\ n_{p(j)}\cdot\mu(\D) & \leq |\Gamma_i|\cdot\|\Gama\|_{0,\infty}\cdot\mu(\D).
	\end{align}
	Using both bounds, we get
	\begin{align*}
	\left| \sum_{t\in \mathcal{T}}\Gamma_t\d_t^T\d_i \right|
	&\geq |\Gamma_i| -  |\Gamma_i| \cdot (\|\Gama\|_{0,\infty}-1)\cdot\mu(\D)\\
	&>|\Gamma_i|\cdot\|\Gama\|_{0,\infty}\mu(\D)
	\geq\max_{j\notin\mathcal{T}}\left| \sum_{t\in \mathcal{T}}\Gamma_t\d_t^T\d_j \right|.
	\end{align*}
	Thus,
	\begin{equation}
	1-(\|\Gama\|_{0,\infty}-1)\cdot\mu(\D) > \|\Gama\|_{0,\infty}\cdot\mu(\D).
	\end{equation}
	From this we obtain the requirement stated in the theorem. Thus, this condition guarantees the success of the first OMP step, implying it will choose an atom inside the true support.
	
	The next step in the OMP algorithm is an update of the residual. This is done by decreasing a term proportional to the chosen atom (or atoms within the correct support in subsequent iterations) from the signal. Thus, this residual is also a linear combination of the same atoms as the original signal. As a result, the $\Loi$ norm of the residual's representation is less or equal than the one of the true sparse code $\Gama$. Using the same set of steps we obtain that the condition on the $\Loi$ norm \eqref{eq:Loi_condition} guarantees that the algorithm chooses again an atom from the true support of the solution. Furthermore, the orthogonality enforced by the least-squares step guarantees that the same atom is never chosen twice. As a result, after $\|\Gama\|_0$ iterations the OMP will find all the atoms in the correct support, reaching a residual equal to zero.
\end{proof}

\subsection{BP Success Guarantee (Proof of Theorem \ref{thm:BPSuccess})}
\begin{proof}
	Define the following set
	\begin{equation}
	\C=\set*{ \hat{\Gama} \given
		\begin{gathered}
		\begin{split}
		&&\hat{\Gama}&\neq\Gama,& \quad \D(\hat{\Gama}-\Gama)&=\mathbf{0}\\
		&&\|\hat{\Gama}\|_1&\leq\|\Gama\|_1,& \quad \|\hat{\Gama}\|_{0,\infty}&>\|\Gama\|_{0,\infty}
		\end{split}
		\end{gathered}
	}.
	\end{equation}
	This set contains all alternative solutions which have lower or equal $\ell_1$ norm and higher $\|\cdot\|_{0,\infty}$ norm. If this set is non-empty, the solution of the basis pursuit is different from $\Gama$, implying failure. In view of our uniqueness result, and the condition posed in this theorem on the $\Loi$ cardinality of $\Gama$, every solution $\hat{\Gama}$ which is not equal to $\Gama$ must have a higher $\|\cdot\|_{0,\infty}$ norm. Thus, we can omit the requirement $\|\hat{\Gama}\|_{0,\infty}>\|\Gama\|_{0,\infty}$ from $\C$.
	
	By defining $\Delt=\hat{\Gama}-\Gama$, we obtain a shifted version of the set,
	\begin{equation}
	\C_s=\set*{ \Delt \given
		\begin{gathered}
		\Delt\neq\mathbf{0}, \quad \D\Delt=\mathbf{0}\\
		\mathbf{0}\geq\|\Delt+\Gama\|_1-\|\Gama\|_1
		\end{gathered}
	}.
	\end{equation}
	In what follows, we will enlarge the set $\C_s$ and prove that it remains empty even after this expansion.
	Since $\D\Delt=\mathbf{0}$, then $	\D^T\D\Delt=\mathbf{0}$. By subtracting $\Delt$ from both sides, we obtain
	\begin{equation}
	-\Delt=(\D^T\D-\mathbf{I})\Delt.
	\label{Eq:equality_constraint_L1proof}
	\end{equation}
	Taking an entry-wise absolute value on both sides, we obtain
	\begin{equation}
	|\Delt|=|(\D^T\D-\mathbf{I})\Delt|\leq|\D^T\D-\mathbf{I}|\cdot|\Delt|,
	\end{equation}
	where we have applied the triangle inequality to the multiplication of the $i^{th}$ row of $(\D^T\D-\mathbf{I})$ by the vector $\Delt$. Note that in the convolutional case $\D^T\D$ is zero for inner products of atoms which do not overlap. Furthermore, the $i^{th}$ row of $\D^T\D$ is non-zero only in the indices which correspond to the stripe that fully contains the $i^{th}$ atom, and these non-zero entries can be bounded by $\mu(\D)$. Thus, extracting the $i^{th}$ row from the above equation gives
	\begin{equation*}
	|\Delta_i|\leq\mu(\D)\left(\|\delt_{p(i)}\|_1-|\Delta_i|\right), \label{Eq:Delta_i_abs}
	\end{equation*}
	where $p(i)$ is the stripe centered around the $i^{th}$ atom and $\delt_{p(i)}$ is the corresponding sparse vector of length $(2n-1)m$ extracted from $\Delt$, as can be seen in Figure \ref{proofL1}.
	\begin{figure}[t]
		\centering
		\includegraphics[trim =10 0 10 0 ,width=.33\textwidth]{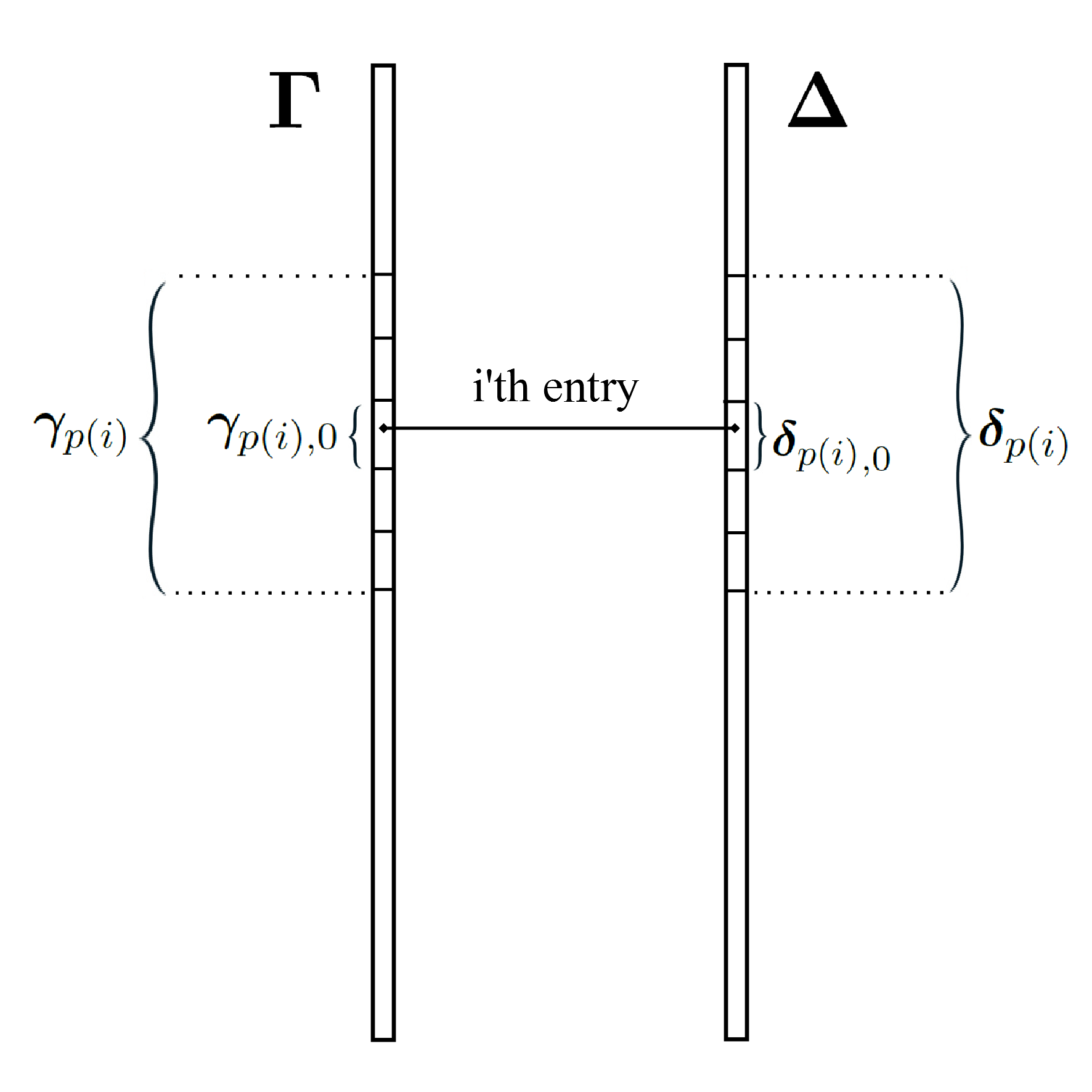}
		\caption{On the left we have the global sparse vector $\Gama$, a stripe  $\gama_{p(i)}$ (centered around the $i^{th}$ atom) extracted from it, and the center of this stripe $\gama_{p(i),0}$. The length of the stripe $\gama_{p(i)}$ is $(2n-1)m$ and the length of $\gama_{p(i),0}$ is $m$. On the right we have the corresponding global vector $\Delt$. Notice that if we were to consider the $i+1$ entry instead of the $i^{th}$, the vector corresponding to $\delt_{p(i)}$ would not change because the atoms $i$ and $i+1$ are fully overlapping.}
		\label{proofL1}
		\vspace{-0.2cm}
	\end{figure}
	This can be written as
	\begin{equation}
	|\Delta_i|\leq\frac{\mu(\D)}{\mu(\D)+1}\|\delt_{p(i)}\|_1.
	\end{equation}
	The above expression is a relaxation of the equality in Equation \eqref{Eq:equality_constraint_L1proof}, since each entry $\Delta_i$ is no longer constrained to a specific value, but rather bounded from below and above. Therefore, by putting the above into $\C_s$, we obtain a larger set $\C_s^1$:
	\begin{equation} 
	\C_s\subseteq\C_s^1=\set*{ \Delt \given
		\begin{gathered}
		\Delt\neq\mathbf{0}, \quad 	\mathbf{0}\geq\|\Delt+\Gama\|_1-\|\Gama\|_1 \\
		|\Delta_i|\leq\frac{\mu(\D)}{\mu(\D)+1}\|\delt_{p(i)}\|_1, \quad\forall i
		\end{gathered}
	}.
	\end{equation}
	Next, let us examine the second requirement
	\begin{align}
	\mathbf{0}\geq & \|\Delt+\Gama\|_1-\|\Gama\|_1 \\ \label{Eq:equality_constraint_L1proof_2}
	= & \sum_{i\in \mathcal{T}(\Gama)} \left(|\Delta_i+\Gamma_i|-|\Gamma_i|\right) + \sum_{i\notin \mathcal{T}(\Gama)} |\Delta_i|,
	\end{align}
	where, as before, $\mathcal{T}(\Gama)$ denotes the support of $\Gama$.
	Using the reverse triangle inequality, $|a+b|-|b|\geq-|a|$, we obtain
	\begin{align}\label{eq:inequality_bp_proof}
	\mathbf{0}&\geq\sum_{i\in \mathcal{T}(\Gama)} \left(|\Delta_i+\Gamma_i|-|\Gamma_i|\right) + \sum_{i\notin \mathcal{T}(\Gama)} |\Delta_i|\\ \nonumber
	&\geq\sum_{i\in \mathcal{T}(\Gama)} -|\Delta_i| + \sum_{i\notin \mathcal{T}(\Gama)} |\Delta_i| =\|\Delt\|_1-2\mathds{1}^T_{\mathcal{T}(\Gama)}|\Delt|, 
	\end{align}
	where the vector $\mathds{1}_{\mathcal{T}(\Gama)}$ contains ones in the entries corresponding to the support of $\Gama$ and zeros elsewhere. 
	Note that every vector satisfying Equation \eqref{Eq:equality_constraint_L1proof_2} will necessarily satisfy Equation \eqref{eq:inequality_bp_proof}. Therefore, by relaxing this constraint in $\C_s^1$, we obtain a larger set $\C_s^2$
	\begin{equation} \label{eq:CS2}
	\C_s^1\subseteq\C_s^2=\set*{ \Delt \given
		\begin{gathered}
		\Delt\neq\mathbf{0}, \quad 	\mathbf{0}\geq\|\Delt\|_1-2\mathds{1}^T_{\mathcal{T}(\Gama)}|\Delt|\\
		|\Delta_i|\leq\frac{\mu(\D)}{\mu(\D)+1}\|\delt_{p(i)}\|_1, \quad\forall i
		\end{gathered}
	}.
	\end{equation}
	Next, we will show the above defined set is empty for a small-enough support. We begin by summing the inequalities \mbox{$|\Delta_i|\leq\frac{\mu(\D)}{\mu(\D)+1}\|\delt_{p(i)}\|_1$} over the support of $\gama_{{p(i)},0}$. Recall that $\gama_{p(i)}$ is defined to be a stripe of length $(2n-1)m$ extracted from the global representation vector and $\gama_{p(i),0}$ corresponds to the central $m$ coefficients in the $p(i)$ stripe. Also, note that $\delt_{p(i)}$ is equal for all the entries inside the support of $\gama_{p(i),0}$. Since all the atoms inside the support of $\gama_{p(i),0}$ are fully overlapping, $\delt_{p(i)}$ does not change, as explained in Figure \ref{proofL1}.
	Thus, we obtain
	\begin{equation*}
	\mathds{1}^T_{\mathcal{T}(\gama_{p(i),0})}|\Delt|\leq\frac{\mu(\D)}{\mu(\D)+1}\cdot\|\gama_{p(i),0}\|_0\cdot\|\delt_{p(i)}\|_1.
	\end{equation*}
	Summing over all different $p(i)$ we obtain
	\begin{equation} \label{eq:goingback}
	\mathds{1}^T_{\mathcal{T}(\Gama)}|\Delt|\leq\frac{\mu(\D)}{\mu(\D)+1}\sum_k\|\gama_{k,0}\|_0\cdot\|\delt_k\|_1.
	\end{equation}
	Notice that in the sum above we multiply the $\ell_0$-norm of the \emph{local sparse vector} $\gama_{k,0}$ by the $\ell_1$ norm of the \emph{stripe} $\delt_k$. In what follows, we will show that, instead, we could multiply the $\ell_0$-norm of the \emph{stripe} $\gama_k$ by the $\ell_1$ norm of the \emph{local sparse vector} $\delt_{k,0}$, thus changing the order between the two. As a result, we will obtain the following inequality:
	\begin{equation}
	\mathds{1}^T_{\mathcal{T}(\Gama)}|\Delt|\leq\frac{\mu(\D)}{\mu(\D)+1}\sum_k\|\gama_{k}\|_0\cdot\|\delt_{k,0}\|_1.
	\end{equation}	
	Returning to Equation \eqref{eq:goingback}, we begin by decomposing the $\ell_1$ norm of the stripe $\delt_k$ into all possible shifts ($m-$dimensional chunks) and pushing the sum outside, obtaining:
	\begin{align}
	\mathds{1}^T_{\mathcal{T}(\Gama)}|\Delt|&\leq\frac{\mu(\D)}{\mu(\D)+1}\sum_k\|\gama_{k,0}\|_0\cdot\|\delt_k\|_1\\
	&=\frac{\mu(\D)}{\mu(\D)+1}\sum_k\left(\|\gama_{k,0}\|_0\sum_{j=k-n+1}^{k+n-1}\|\delt_{j,0}\|_1\right)\\
	&=\frac{\mu(\D)}{\mu(\D)+1}\sum_k\sum_{j=k-n+1}^{k+n-1}\|\gama_{k,0}\|_0\|\delt_{j,0}\|_1. \label{eq:matrix_sum}
	\end{align}
	Define a banded matrix $\A$ (with a band of width $2n-1$) such that $\A_{k,j}=\|\gama_{k,0}\|_0\cdot\|\delt_{j,0}\|_1$,	where $k-n+1\leq j\leq k+n-1$. Notice that the summation in \eqref{eq:matrix_sum} is equal to the sum of all entries in this matrix, where the first sum considers all its rows $k$ while the second sum considers all its columns $j$ (the second sum is restricted to the non-zero band). Instead, this interpretation suggests that we could first sum over all the columns $j$, and only then sum over all the rows $k$ which are inside the band. As a result, we obtain that
	\begin{align}
	\mathds{1}^T_{\mathcal{T}(\Gama)}|\Delt|&\leq\frac{\mu(\D)}{\mu(\D)+1}\sum_k\sum_{j=k-n+1}^{k+n-1}\|\gama_{k,0}\|_0\cdot\|\delt_{j,0}\|_1\\
	&=\frac{\mu(\D)}{\mu(\D)+1}\sum_j\sum_{k=j-n+1}^{j+n-1}\|\gama_{k,0}\|_0\cdot\|\delt_{j,0}\|_1\\
	&=\frac{\mu(\D)}{\mu(\D)+1}\sum_j\left(\|\delt_{j,0}\|_1\sum_{k=j-n+1}^{j+n-1}\|\gama_{k,0}\|_0\right).
	\end{align}
	Summing over all possible shifts we obtain the $\ell_0$-norm of the stripe $\gama_j$; i.e.,
	\begin{equation}
	\mathds{1}^T_{\mathcal{T}(\Gama)}|\Delt| \leq \frac{\mu(\D)}{\mu(\D)+1}\sum_j\|\delt_{j,0}\|_1\cdot\|\gama_{j}\|_0.
	\end{equation}
	Using the definition of $\|\cdot\|_{0,\infty}$
	\begin{align}
	\mathds{1}^T_{\mathcal{T}(\Gama)}|\Delt|&\leq\frac{\mu(\D)}{\mu(\D)+1}\sum_j\|\delt_{j,0}\|_1\cdot\|\gama_{j}\|_0\\
	&\leq\frac{\mu(\D)}{\mu(\D)+1}\sum_j\|\delt_{j,0}\|_1\cdot\|\Gama\|_{0,\infty}\\
	&\leq\frac{\mu(\D)}{\mu(\D)+1}\cdot\|\Delt\|_1\cdot\|\Gama\|_{0,\infty}. \label{eq:second_inequality}
	\end{align}
	For the set $\C_s^2$ to be non-empty, there must exist a $\Delt$ which satisfies
	\begin{align}
	\mathbf{0}\geq & \|\Delt\|_1-2\mathds{1}^T_{\mathcal{T}(\Gama)}|\Delt| \\ \geq & \|\Delt\|_1-2\frac{\mu(\D)}{\mu(\D)+1}\cdot\|\Delt\|_1\cdot\|\Gama\|_{0,\infty},
	\end{align}
	where the first and second inequalities are given in \eqref{eq:inequality_bp_proof} and \eqref{eq:second_inequality}, respectively.
	Rearranging the above we obtain $\|\Gama\|_{0,\infty}\geq\frac{1}{2}\left(1+\frac{1}{\mu(\D)}\right)$. However, we have assumed that $\|\Gama\|_{0,\infty}<\frac{1}{2}\left(1+\frac{1}{\mu(\D)}\right)$ and thus the previous inequality is not satisfied. As a result, the set we have defined is empty, implying that BP leads to the desired solution.
\end{proof}

\renewcommand{\theequation}{C-\arabic{equation}}
\setcounter{equation}{0}  
\section{Properties of the Shifted Mutual Coherence and Stripe Coherence} 
\label{App:MoreOnShiftedMu}
The shifted mutual coherence exhibits some interesting properties: 
\begin{enumerate}[a)]
\item $\mu_s$ is symmetric with respect to the shift $s$, i.e. $\mu_s=\mu_{-s}$.
\item Its maximum over all shifts equals the global mutual coherence of the convolutional dictionary: $\mu(\D) = \underset{s}{\max}\ \mu_s$.
\item The mutual coherence of the local dictionary is bounded by that of the global one: $\mu(\D_L) = \mu_0 \leq\underset{s}{\max}\ \mu_s=\mu(\D)$.
\end{enumerate}
We now briefly remind the definition of the maximal stripe coherence, as we will make use of it throughout the rest of the appendix.
Given a vector $\Gama$, recall that the stripe coherence is defined as $\zeta(\gama_i) = \sum_{s=-n+1}^{n-1} n_{i,s}\ \mu_s$, where $n_{i,s}$ is the number of non-zeros in the $s^{th}$ shift of $\gama_i$, taken from $\Gama$. The reader might ponder how the maximal stripe coherence might be computed.
Let us now define the vector $\mathbf{v}$ which contains in its $i^{th}$ entry the number $n_{i,0}$. Using this definition, the coherence of every stripe can be calculated efficiently by convolving the vector $\mathbf{v}$ with the vector of the shifted mutual coherences $[\mu_{-n+1},\dots,\mu_{-1},\mu_0,\mu_1,\dots,\mu_{n-1}]$.

Next, we provide an experiment in order to illustrate the shifted mutual coherence. To this end, we generate a random local dictionary with $m=8$ atoms of length $n=64$ and afterwards normalize its columns. We then construct a convolutional dictionary which contains global atoms of length $N = 640$. We exhibit the shifted mutual coherences for this dictionary in Figure \ref{fig:Mus}. 

\begin{figure}[t]
\begin{center}
	\begin{subfigure}[t]{.24\textwidth}
		\includegraphics[trim = 0 0 50 20, width=\textwidth]{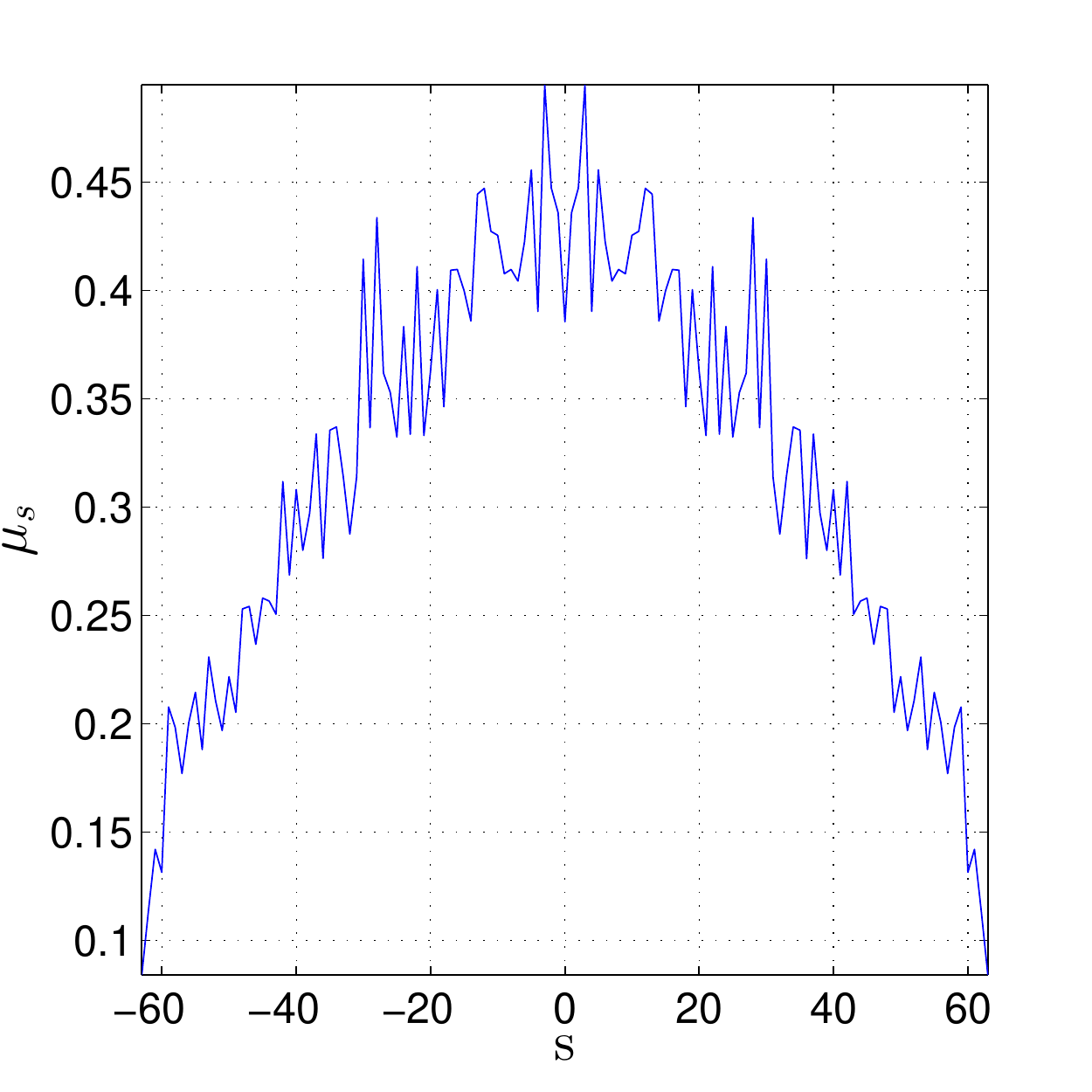}
		\caption{}
		\label{fig:Mus}
	\end{subfigure}
	\begin{subfigure}[t]{.24\textwidth}
		\includegraphics[trim = 50 0 0 50, width=\textwidth]{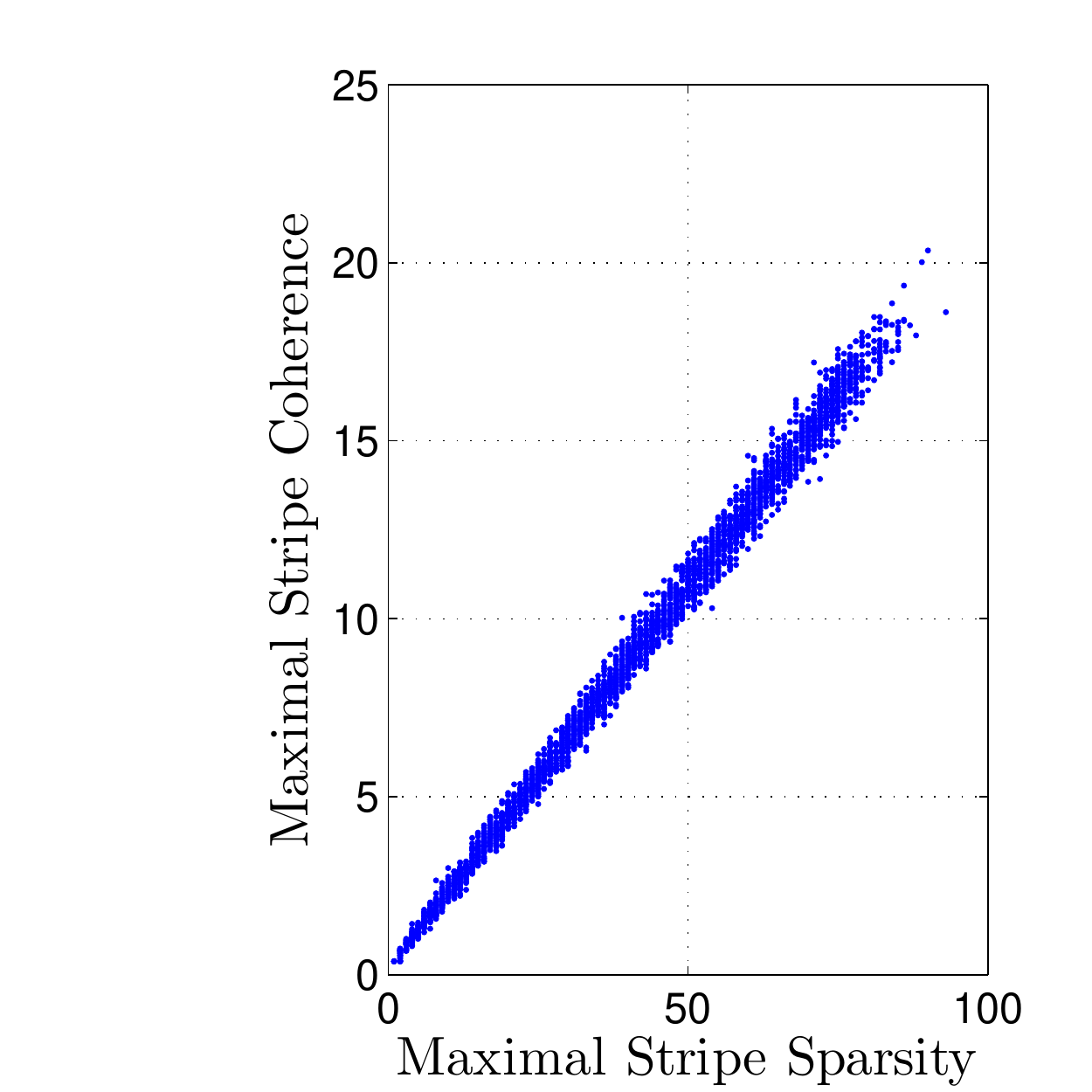}
		\caption{}
		\label{fig:L0inf_sc}
	\end{subfigure}
	\caption{Left: the shifted mutual coherence as function of the shift. The larger the shift between the atoms, the lower $\mu_s$ is expected to be. Right: the maximal stripe coherence as a function of the $\Loi$ norm, for random realizations of global sparse vectors.}
	\vspace{-0.5cm}
\end{center}
\end{figure}

Given this dictionary, we generate sparse vectors with random supports of cardinalities in the range $\left[1,300\right]$. For each sparse vector we compute its $\Loi$ norm by searching for the densest stripe, and its maximal stripe coherence using the convolution mentioned above. In Figure \ref{fig:L0inf_sc} we illustrate the connection between the $\Loi$ norm and the maximal stripe coherence for this set of sparse vectors. As expected, the $\Loi$ norm and the maximal stripe coherence are highly correlated. In Appendix \ref{App:OMPproofStripeCoherence}, we will show an analysis which is based on both these measures. Although the theorems based on the stripe coherence are sharper, they are harder to comprehend. In this experiment we attempted to alleviate this by showing an intuitive connection between the two.

We now present a theorem relating the stripe coherences of related sparse vectors.
\begin{thm} \label{Thm:max_stripe_coherence_contained}
Let $\Gama_1$ and $\Gama_2$ be two global sparse vectors such that the support of $\Gama_1$ is contained in the support of $\Gama_2$. Then the maximal stripe coherence of $\Gama_1$ is less or equal than the maximal stripe coherence of $\Gama_2$.
\end{thm}

\begin{proof}
Denote by $\gama_i^1$ and $\gama_i^2$ the $i^{th}$ stripe extracted from $\Gama_1$ and $\Gama_2$, respectively. Also, denote by $n_{i,s}^{1}$ and $n_{i,s}^{2}$ the number of non-zeros in the $s^{th}$ shift of $\gama_i^1$ and $\gama_i^2$, respectively. Since the support of $\Gama_1$ is contained in the support of $\Gama_2$, we have that $\forall i, s \quad n_{i,s}^{1}\leq n_{i,s}^{2}$. As a result, we have that
\begin{equation}
\max_i\sum_{s=-n+1}^{n-1} n_{i,s}^{1} \mu_s\leq
\max_i\sum_{s=-n+1}^{n-1} n_{i,s}^{2} \mu_s.
\end{equation}
The left-hand side of the above inequality is the maximal stripe coherence of $\Gama_1$, while the right-hand side is the corresponding one for $\Gama_2$. Thus, we conclude that the maximal stripe coherence of $\Gama_1$ is less or equal than the maximal stripe coherence of $\Gama_2$.
\end{proof}

\renewcommand{\theequation}{D-\arabic{equation}}
\setcounter{equation}{0}  
\section{OMP Success Guarantee via Stripe Coherence (Proof of Theorem \ref{thm:OMPSuccess_stripeCoherence})} 
\label{App:OMPproofStripeCoherence}

\begin{proof}
	
	The first steps of this proof are exactly those derived in proving Theorem \ref{thm:OMPSuccess}, and thus we omit them for the sake brevity. Recall that in order for the first step of OMP to succeed, we require
	\begin{equation} \label{eq:inequality2}
	\left| \sum_{t\in \mathcal{T}}\Gamma_t\d_t^T\d_i \right| > \max_{j\notin\mathcal{T}} \left| \sum_{t\in \mathcal{T}}\Gamma_t\d_t^T\d_j \right|.
	\end{equation}
	Lower bounding the left hand side of the above inequality, we can write
	\begin{equation}
	\left| \sum_{t\in \mathcal{T}}\Gamma_t\d_t^T\d_i \right| \geq |\Gamma_i| - |\Gamma_i|\sum_{t \in \mathcal{T}_{p(i)},t\neq i}|\d_t^T\d_i |,
	\end{equation}
	as stated previously in Equation \eqref{Eq:EquationFromOMPproof}.
	Instead of summing over the support $\mathcal{T}_{p(i)}$, we can sum over all the supports $\mathcal{T}_{{p(i)},s}$, which correspond to all possible shifts. We can then write
	\begin{equation}
	\left| \sum_{t\in \mathcal{T}}\Gamma_t\d_t^T\d_i \right| \geq |\Gamma_i| -  |\Gamma_i|  \sum_{s = -n+1}^{n-1} \ \sum_{\substack{t \in \mathcal{T}_{{p(i)},s} \\ t\neq i}} |\d_t^T\d_i |.
	\end{equation}
	We can bound the right term by using the number of non-zeros in each sub-support $\mathcal{T}_{{p(i)},s}$, denoted by $n_{{p(i)},s}$, together with the corresponding shifted mutual coherence $\mu_s$. Also, we can disregard the constraint $t\neq i$ in the above summation by subtracting an extra $\mu_0$ term, obtaining:
	\begin{equation}
	\left| \sum_{t\in \mathcal{T}}\Gamma_t\d_t^T\d_i \right| \geq |\Gamma_i| -  |\Gamma_i| \left( \sum_{s = -n+1}^{n-1}\mu_s n_{p(i),s}-\mu_0 \right).
	\end{equation}
	Bounding the above by the maximal stripe coherence, we obtain
	\begin{equation}
	\left| \sum_{t\in \mathcal{T}}\Gamma_t\d_t^T\d_i \right| \geq |\Gamma_i| -  |\Gamma_i|  \left( \max_k\sum_{s = -n+1}^{n-1}\mu_s n_{k,s}-\mu_0 \right).
	\end{equation}
	In order to upper bound the right hand side of Equation \eqref{eq:inequality2} we follow the steps leading to Equation \eqref{eq:SkipSteps2}, resulting in
	\begin{equation}
	\max_{j\notin\mathcal{T}}\left| \sum_{t\in \mathcal{T}}\Gamma_t\d_t^T\d_j \right| \leq |\Gamma_i|\max_{j\notin\mathcal{T}}\sum_{t \in \mathcal{T}_{p(j)}}|\d_t^T\d_j |.
	\end{equation}
	Using a similar decomposition of the support and the definition of the shifted mutual coherence, we have
	\begin{align}
	\max_{j\notin\mathcal{T}}\left| \sum_{t\in \mathcal{T}}\Gamma_t\d_t^T\d_j \right| &\leq |\Gamma_i|\max_{j\notin\mathcal{T}} \sum_{s=-n+1}^{n-1} \ \sum_{t\in\mathcal{T}_{p(j),s}} |\d_t^T\d_j |\\
	& \leq |\Gamma_i|\max_{j\notin\mathcal{T}} \sum_{s=-n+1}^{n-1} \mu_s n_{p(j),s}.
	\end{align}
	Once again bounding this expression by the maximal stripe coherence, we obtain
	\begin{equation}
	\max_{j\notin\mathcal{T}}\left| \sum_{t\in \mathcal{T}}\Gamma_t\d_t^T\d_j \right| \leq |\Gamma_i|\cdot\max_k\sum_{s=-n+1}^{n-1} \mu_s n_{k,s}.
	\end{equation}
	Using both bounds, we have that
	\begin{align*}
	\left| \sum_{t\in \mathcal{T}}\Gama_t \d_t^T\d_i \right|
	& \geq|\Gamma_i| -  |\Gamma_i| \left( \max_k\sum_{s=-n+1}^{n-1} \mu_s n_{k,s}-\mu_0 \right) \\
	& >|\Gamma_i|\cdot \max_k\sum_{s=-n+1}^{n-1} \mu_s n_{k,s} \\
	& \geq\max_{j\notin\mathcal{T}}\left| \sum_{t\in \mathcal{T}}\Gamma_t\d_t^T\d_j \right|.
	\end{align*}
	Thus,
	\begin{equation}
	1-\max_k\sum_{s=-n+1}^{n-1} \mu_s n_{k,s}+\mu_0 \ > \ \max_k\sum_{s=-n+1}^{n-1} \mu_s n_{k,s}.
	\end{equation}
	Finally, we obtain
	\begin{equation}
	\max_k\ \zeta_k=\max_k\sum_{s=-n+1}^{n-1} \mu_s n_{k,s}\ < \ \frac{1}{2}\left(1+\mu_0\right),
	\end{equation}
	which is the requirement stated in the theorem. Thus, this condition guarantees the success of the first OMP step, implying it will choose an atom inside \mbox{the true support $\mathcal{T}$.}
		
	The next step in the OMP algorithm is an update of the residual. This is done by decreasing a term proportional to the chosen atom (or atoms within the correct support in subsequent iterations) from the signal. Thus, the support of this residual is contained within the support of the true signal. As a result, according to Theorem \ref{Thm:max_stripe_coherence_contained}, the maximal stripe coherence corresponding to the residual is less or equal to the one of the true sparse code $\Gama$. Using the same set of steps we obtain that the condition on the maximal stripe coherence \eqref{eq:stripe_coherence_condition3} guarantees that the algorithm chooses again an atom from the true support of the solution. Furthermore, the orthogonality enforced by the least-squares step guarantees that the same atom is never chosen twice. As a result, after $\|\Gama\|_0$ iterations the OMP will find all the atoms in the correct support, reaching a residual equal to zero.
\end{proof}

We have provided two theorems for the success of the OMP algorithm. Before concluding, we aim to show that assuming $\mu(\D)=\mu_0$, the guarantee based on the stripe coherence is at least as strong as the one based on the $\Loi$ norm. Assume the recovery condition using the $\Loi$ norm is met and as such $\|\Gama\|_{0,\infty} = \max_i \ n_i < \ \frac{1}{2}\left(1+\frac{1}{\mu(\D)}\right)$, where $n_i$ is equal to $\|\gama_i\|_0$. Multiplying both sides by $\mu(\D)$ we obtain \mbox{$\max_i \ n_i\cdot\mu(\D) < \ \frac{1}{2}\left(1+\mu(\D)\right)$}. Using the above inequality and the properties:
\begin{equation}
1) \sum\limits_{s=-n+1}^{n-1} n_{i,s}=n_i, \qquad 2)\  \forall s \quad \mu_s\leq\mu(\D),
\end{equation}
we have that
\begin{align}
\max_i\sum_{s=-n+1}^{n-1} n_{i,s} \mu_s&\leq\max_i \sum_{s=-n+1}^{n-1} n_{i,s}\mu(\D)\\
&=\max_i \ n_i\cdot\mu(\D) < \ \frac{1}{2}\left(1+\mu(\D)\right).
\end{align}
Thus, we obtain that
\begin{equation}
\max_i\sum_{s=-n+1}^{n-1} n_{i,s} \mu_s < \ \frac{1}{2}\left(1+\mu(\D)\right)=\frac{1}{2}\left(1+\mu_0\right),
\end{equation}
where we have used our assumption that $\mu(\D)=\mu_0$. We conclude that if the recovery condition based on the $\Loi$ norm is met, then so is the one based on the stripe coherence. As a result, the condition based on the stripe coherence is at least as strong as the one based on the $\Loi$ norm.

As a final note, we mention that assuming $\mu(\D)= \underset{s}{\max}\ \mu_s = \mu_0$ is in fact a reasonable assumption. Recall that in order to compute $\mu_s$ we evaluate inner products between atoms which are $s$ indexes shifted from each other. As a result, the higher the shift $s$ is, the less overlap the atoms have, and the less $\mu_s$ is expected to be. Thus, we expect the value $\mu_0$ to be the largest or close to it in most cases.


\bibliographystyle{ieeetr}
\bibliography{MyBib}

\end{document}

%% file: phase_transition_both.tikz
%
%
\begin{tikzpicture}

\begin{axis}[%
width=0.951\figurewidth,
height=\figureheight,
at={(0\figurewidth,0\figureheight)},
scale only axis,
separate axis lines,
every outer x axis line/.append style={black},
every x tick label/.append style={font=\color{black}},
xmin=0,
xmax=100,
xlabel={Maximal Stripe Sparsity $\|\Gama \|_{0,\infty}$},
xmajorgrids,
every outer y axis line/.append style={black},
every y tick label/.append style={font=\color{black}},
ymin=0,
ymax=1,
ylabel={Probability of Success},
ylabel style={align=center,text width=3cm},
ymajorgrids,
axis background/.style={fill=white},
legend style={legend cell align=left,align=left,draw=black}
]

\addplot[color=red,solid,mark=*,mark options={solid}] plot table[row sep=crcr] {%
1	1\\
2	1\\
3	0.999250374812594\\
4	0.998520710059172\\
5	0.997893258426966\\
6	0.999328408327737\\
7	0.998756991920448\\
8	0.998751560549313\\
9	0.997043169722058\\
10	0.994051160023795\\
11	0.99578567128236\\
12	0.996468510888758\\
13	0.997116493656286\\
14	0.995505617977528\\
15	0.995614035087719\\
16	0.996095928611266\\
17	0.996694214876033\\
18	0.994114499732477\\
19	0.995565410199557\\
20	0.994612068965517\\
21	0.997251236943375\\
22	0.992972972972973\\
23	0.99354491662184\\
24	0.994324045407637\\
25	0.993036957686127\\
26	0.990231362467866\\
27	0.99238578680203\\
28	0.993407707910751\\
29	0.991979949874687\\
30	0.992150706436421\\
31	0.991778006166495\\
32	0.990950226244344\\
33	0.992401215805471\\
34	0.990673575129534\\
35	0.988174807197943\\
36	0.988939165409754\\
37	0.989717223650386\\
38	0.991517436380773\\
39	0.99000999000999\\
40	0.988977955911824\\
41	0.988301119023398\\
42	0.982073643410853\\
43	0.986973947895792\\
44	0.986124876114965\\
45	0.987317073170732\\
46	0.985983566940551\\
47	0.991546494281452\\
48	0.986342943854325\\
49	0.981629769194536\\
50	0.97735460627895\\
51	0.970366886171214\\
52	0.953748782862707\\
53	0.9335232668566\\
54	0.912030432715169\\
55	0.862102217936355\\
56	0.82055063913471\\
57	0.741045498547919\\
58	0.667281956622058\\
59	0.559924206537186\\
60	0.470616570327553\\
61	0.386941910705713\\
62	0.305853658536585\\
63	0.226923076923077\\
64	0.155466399197593\\
65	0.108235294117647\\
66	0.0830564784053156\\
67	0.0652372262773723\\
68	0.041018387553041\\
69	0.0202764976958525\\
70	0.0111867704280156\\
71	0.00831024930747922\\
72	0.00648748841519926\\
73	0.0034330554193232\\
74	0.000525762355415352\\
75	0.00113250283125708\\
76	0.000628930817610063\\
77	0.00151285930408472\\
78	0.000956022944550669\\
79	0\\
80	0\\
81	0\\
82	0\\
83	0\\
84	0\\
85	0\\
86	0\\
87	0\\
88	0\\
89	0\\
90	0\\
91	0\\
92	0\\
93	0\\
94	0\\
95	0\\
96	0\\
97	0\\
98	0\\
99	0\\
100	0\\
};
\addlegendentry{\footnotesize{Experimental Data BP}};

\addplot[color=blue,solid,mark=o,mark options={solid}] plot table[row sep=crcr] {%
1	1\\
2	1\\
3	1\\
4	1\\
5	1\\
6	1\\
7	1\\
8	1\\
9	1\\
10	1\\
11	1\\
12	1\\
13	1\\
14	1\\
15	1\\
16	1\\
17	1\\
18	1\\
19	1\\
20	1\\
21	1\\
22	1\\
23	1\\
24	1\\
25	1\\
26	1\\
27	1\\
28	1\\
29	0.999478079331942\\
30	1\\
31	0.99897066392177\\
32	0.999479708636837\\
33	0.999003984063745\\
34	0.998980112187659\\
35	0.999491353001017\\
36	0.998481781376518\\
37	0.995529061102832\\
38	0.993496748374187\\
39	0.991826923076923\\
40	0.987593052109181\\
41	0.985074626865672\\
42	0.986875315497224\\
43	0.974531475252283\\
44	0.962776659959759\\
45	0.963321446765155\\
46	0.947072599531616\\
47	0.937806072477963\\
48	0.926133469179827\\
49	0.901277013752456\\
50	0.875121477162293\\
51	0.84952380952381\\
52	0.813233223838573\\
53	0.772528007793473\\
54	0.753664302600473\\
55	0.683736367946894\\
56	0.654767726161369\\
57	0.608381502890173\\
58	0.52144578313253\\
59	0.515503875968992\\
60	0.438036224976168\\
61	0.389518413597734\\
62	0.325969563082965\\
63	0.287747839349263\\
64	0.246468582562104\\
65	0.196759259259259\\
66	0.170582706766917\\
67	0.120941176470588\\
68	0.0865518845974872\\
69	0.071360153256705\\
70	0.0542778288868445\\
71	0.0386847195357834\\
72	0.0219728845254792\\
73	0.0175867122618466\\
74	0.0132382892057026\\
75	0.013172338090011\\
76	0.0051150895140665\\
77	0.00593471810089021\\
78	0.00552995391705069\\
79	0.00238095238095238\\
80	0.00470957613814757\\
81	0.00210526315789474\\
82	0.00335570469798658\\
83	0\\
84	0\\
85	0\\
86	0\\
87	0\\
88	0\\
89	0\\
90	0\\
91	0\\
92	0\\
93	0\\
94	0\\
95	0\\
96	0\\
97	0\\
98	0\\
99	0\\
100	0\\
};
\addlegendentry{\footnotesize{Experimental Data OMP}};

\addplot [color=black,solid,forget plot]
  table[row sep=crcr]{%
0	0\\
100	0\\
};
\addplot [color=black,dashed,line width=2.0pt]
  table[row sep=crcr]{%
6.05373654502998	0\\
6.05373654502998	1\\
};
\addlegendentry{\footnotesize{Theoretical Bound}};

\end{axis}
\end{tikzpicture}%